\documentclass[graybox]{svmult}

\usepackage{mathptmx}       % selects Times Roman as basic font
\usepackage{helvet}         % selects Helvetica as sans-serif font
\usepackage{courier}        % selects Courier as typewriter font
\usepackage{type1cm}        % activate if the above 3 fonts are
\usepackage{amsmath}
\usepackage{amssymb}
\usepackage{makeidx}         % allows index generation
\usepackage{graphicx}        % standard LaTeX graphics tool
\usepackage{multicol}        % used for the two-column index
\usepackage[bottom]{footmisc}% places footnotes at page bottom
\usepackage{amsbsy}
\usepackage{enumitem}

\newcommand{\hu}{\tilde{u}}
\newcommand{\hv}{\tilde{v}}

\makeindex             % used for the subject index
\begin{document}

\title*{Hamiltonian field theory close to the wave equation: from Fermi-Pasta-Ulam to water waves}
\titlerunning{Hamiltonian field theory: from FPU to water waves}
% Use \titlerunning{Short Title} for an abbreviated version of
% your contribution title if the original one is too long
\author{Matteo Gallone and Antonio Ponno}
% Use \authorrunning{Short Title} for an abbreviated version of
% your contribution title if the original one is too long
\institute{Matteo Gallone \at Dipartimento di Matematica ``F. Enriques'', Università degli Studi di Milano, via Cesare Saldini 50, 20133 Milano, Italy, \email{matteo.gallone@unimi.it}
\and Antonio Ponno \at Dipartimento di Matematica ``T. Levi-Civita'', Università degli Studi di Padova, via Trieste 63, 35121 Padova, Italy,  \email{ponno@math.unipd.it}
}
%
% Use the package "url.sty" to avoid
% problems with special characters
% used in your e-mail or web address
%
\maketitle

\abstract{In the present work we analyse the structure of the Hamiltonian field theory in the neighbourhood of the wave equation $q_{tt}=q_{xx}$. We show that, restricting to ``graded'' polynomial perturbations in $q_x$, $p$ and their space derivatives of higher order, the local field theory is equivalent, in the sense of the Hamiltonian normal form, to that of the Korteweg-de Vries hierarchy of second order. Within this framework, we explain the connection between the theory of water waves and the Fermi-Pasta-Ulam system.}

\section{Introduction}

The present work aims to treat the perturbations of a linear string in the framework of classical Hamiltonian field theory. The unperturbed base model we have in mind, the linear string, is described by the one-dimensional wave equation
\begin{equation}
\label{eq:kg}
q_{tt}\;=\;c^2q_{xx} \, ,
\end{equation}  
where $q:\mathbb{R}\times D\to\mathbb{R}:(t,x)\to q(t,x)$ is the unknown, real-valued field, and $c$ is a real, positive parameter, the \emph{speed of the wave}. As usual, partial derivatives are denoted by subscripts, i.e. $q_t=\partial_t q$, $q_x=\partial_x q$, and so on. Concerning the space domain $D$ and the boundary conditions of the field $q$, we here focus on the $1$-periodic case, namely $D=\mathbb{T}:=\mathbb{R}/\mathbb{Z}$ (the $L$-periodic case, with $D=\mathbb{R}/(L\mathbb{Z})$, can always be reduced to the case $L=1$ by rescaling both the independent variables to $x'= x/L$, $t'=t/L$). 
  
Solving equation (\ref{eq:kg}), for any initial condition
$q(0,x)$, $q_t(0,x)$ defined on $\mathbb{T}$ and regular enough, is a standard exercise in Fourier analysis. Indeed, substituting
$q(t,x)=\sum_{k\in\mathbb{Z}}\hat{q}_k(t)e^{\imath2\pi k x}$ ($\imath$ is the imaginary unit) into (\ref{eq:kg}), one gets
\[
\frac{d^2\hat{q}_k}{dt^2}=-4\pi^2c^2k^2 \, \hat{q}_k\ ,
\]
which implies $\hat{q}_k(t)=a_ke^{\imath\omega_kt}+\bar{a}_{-k}e^{-\imath\omega_kt}$, where the $a_k$ are complex constants (the bar denoting complex conjugation), and
\begin{equation}
\label{eq:omk}
\omega_k:=2\pi c |k|\ ;\ \ k\in\mathbb{Z}\ .
\end{equation}
Observe that $\omega_{-k}=\omega_k$, which implies $\overline{\hat{q}}_k=\hat{q}_{-k}$, i.e. $q$ is real.
Relation (\ref{eq:omk}) defines the dispersion relation of the wave equation. A given space periodic system,
characterised by a certain dispersion relation $k\to\omega_k$ is said to be \emph{non dispersive} if $\omega_{k+1}-\omega_k$ is piecewise constant, i.e. if $\omega_k$ is piecewise linear in $k$ and this is clearly the case for the wave equation. One can check that the solution $q(t,x)$ of the problem is time periodic for all initial conditions, the period being $2\pi/\omega_1=1/c$.

It is almost impossible to give a complete account of physical phenomena that, to the first linear approximation, are described by the wave equation. Let us just mention, to have in mind concrete examples that we are going to analyse later, wave propagation in fluids and long-wavelength vibrations of interacting particle chains. In all these problems, the need to go beyond the first approximation arises, in order to take into account the effects of both nonlinearity and dispersion, typically determining whether some interesting form of energy localisation may take place, as opposed to a fast energy spreading among the degrees of freedom of the system. One is thus led to look for a general treatment of the possible perturbations of equation (\ref{eq:kg}) regardless of the specific physical problem giving rise to it. This in turn calls for the restriction to a mathematical context where the possible perturbations constitute a well-defined ordered class of objects. We do this within the framework of Hamiltonian field theory, at the price to exclude, among others, all the dissipative effects from the theory (no claim is made here about their irrelevance: the other way around. See for example the enlightening discussion made by Nekhoroshev in \cite{Nek}). 
Moreover, we consider nonlinear and dispersive perturbations depending on $q_x$, $p$ and their higher order derivatives, but not on $q$. Indeed, all systems made of interacting particles, such as solids, fluids and gasses, in absence of external forces, and on a sufficiently large space scale, are described by a certain wave equation at the linear level, with perturbations depending, in principle, only by the space derivatives of the field (and its momentum, possibly). This is due to the fact that interactions in matter depend on \emph{differences} of coordinates, which in the continuum approximation correspond to \emph{derivatives}. 

On the other hand, considering smooth perturbations of the wave equation depending on $q$ (not only through derivatives) would be interesting as well. For example, as shown by Bambusi and Nekhoroshev and by Nekhoroshev \cite{Nek,BamNek1,BamNek2}, the smooth perturbations of the wave equation depending on $q$ only (no derivatives) give rise to very nice, long-lasting localisation phenomena. Whether be possible to include such a class of problems in our treatment, drawing meaningful conclusions, looks unclear, at present. 

Although we decided to focus on one-dimensional systems, it is worth mentioning that the techniques presented here can be generalised to study problems in higher space dimension. In this case one can predict, for example, energy localisation for a certain class of anisotropic rectangular lattices \cite{GaPa}.

\medskip

The paper is organised as follows. In Section 2 we introduce the Hamiltonian formalism of classical field theory, at the end of which we provide an informal presentation of the main results. Section 3 contains the elements of perturbation theory framed in the more general context of Poisson systems, which is the one appropriate to our purposes. Section 4 contains the formal statements and proofs of the results. The application of such results to the FPU problem and to the water wave problem is treated in Section 5. Finally, a short list of open problems is provided in Section 6.

\section{Outline of the method and results}

\subsection{Hamiltonian field theory}

For the sake of completeness, we report here a short review on what is meant by \emph{Hamiltonian field theory}. The reader is referred to the monographs \cite{DNF}, \cite{GelFom}, and \cite{MaRa}, for  details and/or a more extensive treatment of the subject.
 
In Hamiltonian field theory the dynamical variables (e.g. coordinates and conjugate momenta) are points in a certain function space, the phase space of the system, and the observables, including the Hamiltonian, are functionals, admitting a density, defined on the phase space. 

In order to specify the notations used below, let us first consider the space of smooth functions, or fields $u:\mathbb{T}\to\mathbb{R}$. A functional $F[u]$, with density $\mathcal{F}$ depending on $x$ and on $u(x)$ and its derivatives up to a given order, is defined as
\begin{equation}
	F[u]=\oint \mathcal{F}(x,u,u_x,u_{xx},\dots) \, dx\ ,
\end{equation}
where here and in the sequel we make use of the short hand notation $\oint:=\int_{\mathbb{T}}$. The functional derivative (or variational derivative) of $F$ with respect to $u$, denoted by $\delta F/\delta u$, is defined by the relation
\begin{equation}
\label{eq:fddef}
\delta{F}[u,\delta u]:=\frac{d}{d\epsilon}F[u+\epsilon \delta u]\big|_{\epsilon=0}=\oint \frac{\delta F}{\delta u}\delta u\ dx\ ,
\end{equation}
for any smooth finite increment $\delta u$ defined on $\mathbb{T}$.  Through repeated integrations by parts and erasing the boundary terms one finds 
\begin{equation}
\label{eq:dFdu}
\frac{\delta F}{\delta u}=\sum_{j\geq0} (-1)^j \frac{d^j}{dx^j} \frac{\partial \mathcal{F}}{\partial (\partial_x^ju)}=
\frac{\partial \mathcal{F}}{\partial u}-\frac{d}{dx}\frac{\partial \mathcal{F}}{\partial u_x}+
\frac{d^2}{dx^2}\frac{\partial \mathcal{F}}{\partial u_{xx}}+\cdots,
\end{equation}
the sum above being finite if $\mathcal{F}$ is a polynomial in $u$ and its derivatives up to a given finite order (as will be in our case). 
Relation (\ref{eq:fddef}) defines the Gateaux, or weak differential of the functional $F$ at $u$ with increment $\delta u$, which under further requirements coincides with the Fr\'echet, or strong differential of $F$; see e.g. \cite{Zorn}. 
The functional derivative is also referred to, in the mathematical literature, as the $L_2$-gradient of $F$ with respect to $u$. Indeed, in the Hilbert space $L_2(\mathbb{T})$ of square integrable functions on 
$\mathbb{T}$, endowed with the usual scalar product $\langle f,g\rangle:=\oint fg\ dx$, one can rewrite (\ref{eq:fddef}) as $\delta F=\langle \delta F/\delta u,\delta u\rangle:=\langle\nabla F,\delta u\rangle$, identical in form to its finite-dimensional counterpart.

In the \emph{Hamiltonian field theory} considered in the present paper, the phase space $\Gamma$ of the system is the space of two components, smooth, real-valued fields $(q(x),p(x))$ defined on $\mathbb{T}$. The observables of the theory are the functionals $F:\Gamma\to\mathbb{R}$ admitting a density $\mathcal{F}$ which is a polynomial in $q(x)$, $p(x)$ and their space derivatives up to a finite order, with coefficients possibly depending on $x$.
One then selects, among the observables, the Hamiltonian defining the given system, namely 
\begin{equation}
\label{eq:Hgen}
H[q,p]:=\oint\mathcal{H}(x,q,p,q_x,p_x,\dots)\ dx\ .
\end{equation}
The motion of the system, a certain curve $\gamma:[t_1,t_2]\ni t\mapsto(q,p)(t)\in\Gamma$, is then specified by a stationary action principle, as in the finite dimensional case. Indeed, defining the \emph{action functional} $S[q,p]$ as
\begin{equation}
\label{eq:action}
S[q,p]:=\int_{t_1}^{t_2}\big[\langle p,q_t\rangle-H\big]\ dt=\int_{t_1}^{t_2}\oint \left[pq_t-\mathcal{H}\right]dt\ dx\ ,
\end{equation}
one defines the actual motion of the system as the critical point of $S$ in the space of smooth curves 
$(q(t,x),p(t,x))$ in $\Gamma$ with fixed ends on the first component: $q(t_1,x):=q_1(x)$, $q(t_2,x):=q_2(x)$, $q_1$ and $q_2$ being two assigned fields on $\mathbb{T}$. The smooth increment curves 
$(\delta q,\delta p)(t)$ must then satisfy the condition $\delta q(t_1,x)=\delta q(t_2,x)=0$.
With the notation just introduced, and performing simple integrations by parts, one gets the differential 
$\delta S$ of the action $S$, namely
\begin{equation}
\label{eq:dS}
\delta S=\int_{t_1}^{t_2}\oint \left[\left(q_t-\frac{\delta H}{\delta p}\right)\delta p-
\left(p_t+\frac{\delta H}{\delta q}\right)\delta q\right]dt\ dx\  .
\end{equation}
This is zero for any increment $(\delta q,\delta p)(t)$ if and only if the following Hamilton equations hold:
\begin{equation}
\label{eq:Heqgen}
q_t=\frac{\delta H}{\delta p}\ \ ;\ \ p_t=-\frac{\delta H}{\delta q}\ .
\end{equation}
This is the Hamilton principle of stationary action in classical field theory.

In this work, we restrict our attention to scalar fields $q$ and $p$ defined on the (flat) unit circle $\mathbb{T}$. However, all the above construction and most of the results presented below can be extended to vector fields defined on any multi-dimensional space domain (not necessarily a torus).
 
Consider now a functional $F[q,p]:=\oint\mathcal{F}(x,q,p,q_x,p_x,\dots)dx$. Its time derivative along the solutions of the Hamilton equations (\ref{eq:Heqgen}) associated to $H$ is computed by means of repeated integrations by parts with respect to
$x$. The result can be written as $dF/dt=\left\{F,H\right\}_{q,p}$, where 
\begin{equation}
\label{eq:Pbra}
\left\{F,H\right\}_{q,p}:=
\oint\left(\frac{\delta F}{\delta q}\frac{\delta H}{\delta p}-\frac{\delta F}{\delta p}\frac{\delta H}{\delta q}\right)\ dx:=\left\langle\nabla F,\mathsf{J}_2\nabla H\right\rangle
\end{equation}
is the \emph{Poisson bracket} of the functionals $F$ and $H$. In the second definition above,
$\mathsf{J}_2:={\small \left(\begin{array}{cc} 0 & 1 \\ -1 & 0 \end{array}\right)}$ is the standard $2\times 2$ symplectic
matrix, $\nabla F= {\small \left(\begin{array}{c} \delta F/\delta q \\ \delta F/\delta p \end{array}\right)}$ and the same for $H$. The product $\xi^T\mathsf{J}_2\eta=\xi_1\eta_2-\xi_2\eta_1$, for any pair of vectors $\xi,\eta\in\mathbb{R}^2$, defines the symplectic 2-form. The Poisson bracket (\ref{eq:Pbra}) defines a bilinear, skew-symmetric product on the algebra of functionals defined on $\Gamma$, and it satisfies the Jacobi identity
$\{\{F,G\}_{q,p},H\}_{q,p}+\{\{G,H\}_{q,p},F\}_{q,p}+\{\{H,F\}_{q,p},G\}_{q,p}\equiv 0$
and the Leibniz rule $\{FG,H\}_{q,p}=F\{G,H\}_{q,p}+\{F,H\}_{q,p}G$ 
for any triple of functionals $F$, $G$, $H$. The algebra of functionals on 
$\Gamma$ endowed with the Poisson bracket becomes a Poisson algebra, and is typically referred to as the algebra of observables. 

\begin{remark}
Given any skew-symmetric bilinear product on an algebra, the Jacobi identity characterises it as a Lie bracket. The latter, by further assuming the Leibniz rule, becomes a Poisson bracket (by definition). Thus, a Poisson algebra is a Lie algebra of Leibniz type.   
\end{remark}

The fundamental Poisson brackets of the Hamiltonian field theory on $\mathbb{T}$ are
\begin{equation}
\{q(x),p(y)\}_{q,p}=\delta(x-y)\ \ ;\ \ \{q(x),q(y)\}_{q,p}=\{p(x),p(y)\}_{q,p}=0\ ,
\end{equation}
where $\delta(x)$ is the Dirac delta distribution on 
$\mathbb{T}$. This is proved by considering the identity $\oint \delta(x-y)f(y)dy=f(x)$, valid for any continuous function on $\mathbb{T}$, from which $\delta f(x)/\delta f(y)=\delta(x-y)$ follows.
As a consequence, the Hamilton equations (\ref{eq:Heqgen}) can be written in the form
\begin{equation}
q_t=\{q,H\}_{q,p}\ \ ;\ \ p_t=\{p,H\}_{q,p}\ .
\end{equation}

\subsection{Results: informal presentation}

Within the Hamiltonian formalism just introduced, we study a well-defined class of problems, defined as follows. We introduce a ``bookkeeping parameter'' $\lambda$ and give a weight $\lambda^2$
to both $q_x$ and $p$, weighting any successive derivative $\partial_x$ of them by $\lambda$. 
Defining $r:=q_x$, this amounts to assume a ``grading'' (perturbative ordering of the dynamical variables
and their derivatives) $r\sim p\ll r_x\sim p_x\ll r_{xx}\sim p_{xx}\dots$, and $(r_x)^2\sim r^3$, where, in a loose notation, $\sim$ and $\ll$ mean ``of the same order of'' and ``of an order smaller than'', respectively. 
For the sake of simplicity,  we assume the smooth density $\mathcal{H}$ of $H$ to be a function of $q_x$, $p$ and their  space derivatives up to order four. Such a limitation is due to the fact that, in the present paper, we do not consider 
$\lambda$-expansions of the Hamiltonian $H$ to degree higher than four, and  with the chosen grading, derivatives of $q_x$
and $p$ of order higher than four enter the perturbative problem from degree five on (in $\lambda$). 
The parameter $\lambda$ is formal: it is necessary to define the grading and to trace the perturbative ordering, and it can be set to one at the end of the computations.

\begin{definition}
\label{def:Class}
The class of problems considered in the present work is defined by the family of Hamiltonians of the form 
\begin{equation}
\label{eq:Hla}
H_\lambda:=\frac{1}{\lambda^4}\oint \mathcal{H}(\lambda^2q_{x},\lambda^2p,\lambda^3q_{xx},\lambda^3p_x,
\dots,\lambda^6q_{xxxxx},\lambda^6p_{xxxx})\ dx\ ,
\end{equation}
with the condition
\begin{equation}
\label{eq:cond}
\left(\frac{\partial^2\mathcal{H}}{\partial q_x^2}\Big|_{\lambda=0}\right)\ 
\left(\frac{\partial^2\mathcal{H}}{\partial p^2}\Big|_{\lambda=0}\right)>0\ .
\end{equation}
\end{definition} 
By Taylor expanding $\mathcal{H}$ in powers of $\lambda$, close to $\lambda=0$, and assuming without loss of generality that $\mathcal{H}|_{(q,p)=0}=0$, one gets a perturbative ordering of the Hamiltonian of the form
\begin{equation}
\label{eq:Hamexp}
H_\lambda=H_0+ \lambda H_1 +\lambda^2 H_2+\lambda^3 H_3 +\lambda^4 H_4 +\cdots\ .
\end{equation}
We here observe that the absence of a term proportional to $1/\lambda^2$ in the latter expansion is due to the conservation of the total momentum $\oint p\ dx$, which can be always set to zero.

The main results are now presented in an informal way, their precise statements and proofs being provided below. The condition (\ref{eq:cond}), which characterises the elliptic nature of the fixed point 
$q=p=0$, implies that there exists a canonical transformation bringing the unperturbed Hamiltonian $H_0$ into the standard wave form
\begin{equation}
K_0:=\oint \frac{p^2+(q_x)^2}{2}\ dx\ ,
\end{equation}
and leaving the perturbative expansion (\ref{eq:Hamexp}) unaltered. The equations of motion associated to the latter Hamiltonian are $q_t=p$, $p_t=q_{xx}$,
i.e., in second order form, the wave equation $q_{tt}=q_{xx}$.

Now, in terms of the variables $r:=q_x$ and $p$, the expanded Hamiltonian (\ref{eq:Hamexp}) reads 
$K_0+\lambda H_1+\lambda^2 H_2+\cdots$, where
$K_0=\frac{1}{2}\oint(p^2+r^2)dx$, and the $H_j$ are functionals whose density is a homogeneous polynomial of ``grade'' $j$ in $r$, $p$ and their derivatives. One then conveniently performs the change of field variables $(r,p)\mapsto (u,v)$ defined by $u=(r+ p)/\sqrt{2}$, $v=(r-p)/\sqrt{2}$, in terms of which 
$K_0=\frac{1}{2}\oint(u^2+v^2)dx$, and its flow separates the left from right wave: $u_t= u_x$, $v_t=-v_x$, so that $u$ and $v$ are simply the left and right translation of the corresponding initial datum, respectively. 

The key idea is now to decouple the left from the right dynamics to higher orders. To such an end, we build up an explicit transformation of the field variables
\[
\mathcal{T}_\lambda:\ (u,v)\mapsto(\tilde{u},\tilde{v})\ ,
\]
$\lambda$-close to the identity, which sets the Hamiltonian 
$H=K_0+\lambda H_1+\lambda ^2H_2+\cdots$ (expressed in the $(u,v)$ variables) into normal form
to order $1\leq s\leq 4$ with respect to $K_0$. This means, by definition, that 
$H\circ\mathcal{T}^{-1}_\lambda=K_0+\lambda Z_1+\lambda ^2Z_2+\cdots$ is such that
the $Z_j$ are first integrals of $K_0$, for $1\leq j\leq s$.

The results proved below are the following. In the general case, i.e. no further hypotheses being added to
the Definition~\ref{def:Class}, we show that the normal form Hamiltonian to order $s=2$ has the form
$K_0+\lambda^2 Z_2+\cdots$, and the corresponding dynamics of the variables $\tilde{u},\tilde{v}$ reads
\begin{equation}
\label{eq:ord2}
\left\{
\begin{split}
	\tilde{u}_t\;&=\; c_l \tilde{u}_x+a_l\kappa_3(\tilde{u})+\cdots	\\
	\tilde{v}_t\;&=\;- c_r \tilde{v}_x-a_r\kappa_3(\tilde{v})+\cdots	
\end{split}
\right.\ .
\end{equation}
On the other hand, in certain relevant cases, such as the ``mechanical'' one, where $\mathcal{H}=
p^2/2+\mathcal{U}(q_x,q_{xx},\dots,q_{xxxxx})$, or that of the water waves, one has $H_1=H_3\equiv0$.
In such situations the normal form Hamiltonian to order $s=4$ has the form $K_0+\lambda^2 Z_2+
\lambda^4 Z_4+\cdots$, whose associated dynamics reads
\begin{equation}
\label{eq:ord4}
\left\{
\begin{split}
	\tilde{u}_t\;&=\; c_l \tilde{u}_x+a_l\kappa_3(\tilde{u})+b_l\kappa_5(\tilde{u})+\cdots	\\
	\tilde{v}_t\;&=\;- c_r \tilde{v}_x-a_r\kappa_3(\tilde{v})-b_r\kappa_5(\tilde{v})+\cdots	
\end{split}
\right.\ .
\end{equation}
In systems (\ref{eq:ord2}) and (\ref{eq:ord4}) $a_{l/r}$, $b_{l/r}$ and $c_{l/r}$ are certain constants
(depending on the model, on the parameter $\lambda$ and on the initial condition), whereas $\kappa_3$ and $\kappa_5$ are the vector fields of the first and second integral in the KdV hierarchy \cite{ArnoldBookCM}, namely
\begin{eqnarray}
	\kappa_3(w)&=&\gamma w w_x+w_{xxx} =\partial_x\frac{\delta I_3}{\delta w} \, ,\\
	\kappa_5(w)&=&\frac{5}{6} \gamma^2 w^2 w_x+\frac{10}{3}\gamma w_x w_{xx} + \frac{5}{3} \gamma w w_{xxx} +  u_{xxxxx} = \partial_x\frac{\delta I_5}{\delta w}\, .\label{eq:kappa5}	
\end{eqnarray}
Here $\gamma \in \mathbb{R}$ is a parameter, whose value is explicitly determined by the first order normal form transformation, whereas the first two integrals $I_3$ and $I_5$ of the KdV hierarchy are given by
\begin{eqnarray}
	I_3&=&\oint \left(\frac{\gamma}{6}w^3-\frac{1}{2}(w_{x})^2 \right) \, dx \, ,\\
	I_5&=&\oint \left(\frac{5\gamma^2}{72} w^4 + \frac{5\gamma}{12} w^2 w_{xx} +\frac{1}{2} (w_{xx})^2 \right) \, dx \, .
\end{eqnarray}
The conclusion is that \emph{both in the general and in the special case, the dynamics of the perturbed wave equation is integrable in the KdV hierarchy sense to the second perturbative order included}. 
\begin{remark} \label{rem:Kodama}
The standard Hamiltonian normal form construction to leading order always leads to (\ref{eq:ord2}). On the other hand, in order to get (\ref{eq:ord4}), the second step of Hamiltonian normalisation is not enough, in general. 
 With the aid of Hamiltonian transformations, we generally succeed in decoupling equations of motion for the two independent variables to higher orders but, in general, this is not enough to conjugate the equations of motion to those of the KdV integrable hierarchy. It is remarkable that, at this point, each of the two decoupled equations of motion fall in a class that was analysed by Kodama \cite{Kod85,Kod87a,Kod87b,Kodama2} (and whose results have been extended to equations on the torus in \cite{GPR}). Without entering the details, which could deserve an entire work, the idea is the following. One starts from a PDE of the form
\begin{equation}
\label{eq:udotgen}
	u_t \;=\; F(u)\;:=\; F_0(u)+\lambda F_1(u)+\lambda^2 F_2(u)+O(\lambda^3)
\end{equation}
and one considers the effect of a change of variables $u \mapsto u+\lambda G(u)$. Denoting with $[\cdot,\cdot]$ the commutator of two vector fields, the effect of the transformation on the RHS of the PDE (\ref{eq:udotgen}) is
\begin{equation}
\begin{split}
	F(u) \mapsto e^{\lambda[G,\cdot]} F(u)=&F_0(u)+\lambda\big(F_1(u)+[G,F_0](u)\big)\\
	+&\lambda^2 \left(F_2(u)+[G,F_1](u)+\frac{1}{2} [G,[G,F_0]](u)\right)+O(\lambda^3) \, .
\end{split}
\end{equation}
The latter conjugation of the vector field $F$ holds in general, i.e. for any $G$. The Kodama transformation consists in making use of the natural grading of the KdV equation in order to choose a $G$ consisting of a finite sum of monomials and satisfying two fundamental requirements. The first one is $[G,F_0]=0$, which allows to leave 
$F_1$ in the KdV hierarchy, as it is given by the normal form construction. The second  one consists just in ``forcing'' $F_2+[G,F_1]$ to fit the KdV hierarchy, even though $F_2$ does not.  This part of the theory is only sketched in the present review and we refer to \cite{Kodama2,GPR} for details.   
\end{remark}
\begin{remark}
The treatment of the general case to orders $s=3$ and $s=4$ requires three and four perturbative steps, respectively, and is currently in progress. 
\end{remark}

\section{Abstract setting: perturbation theory in Poisson systems}

In order to treat our problem, we need to frame our Hamiltonian field theory 
in the more general context of Poisson systems \cite{MaRa,Vaism_Book}. Such a short digression is adapted to our present purposes and does not aim at any generality. 

\subsection{Poisson formalism}

\begin{definition}
	Let $\Gamma$ be the phase space of the system and let $\mathcal{A}(\Gamma)$ be the algebra of real-valued smooth functions defined on $\Gamma$. A binary application, or product, $\{\cdot, \cdot \}:\mathcal{A}(\Gamma) \times \mathcal{A}(\Gamma) \to \mathcal{A}(\Gamma)$ is called a \emph{Poisson bracket} on $\Gamma$ if it satisfies the following properties
	\begin{enumerate}[label=(\roman*)]
		\item Skew-symmetry: $\{F,G\}=-\{G,F\}$;
		\item Left-linearity: $\{\alpha F+ \beta G,H\}=\alpha \{F,H\}+\beta \{G,H\}$;
		\item Jacobi identity: $\{F,\{G,H\}\}+\{G,\{H,F\}\}+\{H,\{F,G\}\}=0$;
		\item Leibniz rule: $\{FG,H\}=F\{G,H\}+\{F,H\}G$,
	\end{enumerate}
$\forall F,G,H \in \mathcal{A}(\Gamma)$ and $\alpha,\beta \in \mathbb{R}$. The pair $(\mathcal{A},\{\ ,\})$
is called \emph{Poisson algebra}.
\end{definition}

\begin{remark}
The bracket $\{\cdot, \cdot\}_{q,p}$ defined in \eqref{eq:Pbra} satisfies axioms (i)-(iv) in the above definition. Thus, the axiomatic definition above contains both the usual Hamiltonian mechanics 
and the field theory (as well as quantum mechanics). 
\end{remark}
For the sake of concreteness, let us consider the case where $\Gamma$ is the space of two components, smooth, real-valued fields $u(x)=(u_1(x),u_2(x))^T$ defined on $\mathbb{T}$ (what we show can be exported to the case of $n$
components, complex-valued fields on a $d$-dimensional domain $D$).
  
By analogy with the standard case (\ref{eq:Pbra}), a bilinear, skew-symmetric, Leibniz bracket on such a space is defined by the formula
\begin{equation}
\label{eq:brascp}
	\{F,G\}_J:= \langle \nabla F, J \nabla G \rangle:=\oint \sum_{i,j=1}^2\frac{\delta F[u]}{\delta u_i}
	J_{ij}[u]\frac{\delta G[u]}{\delta u_j}\ dx \ ,
\end{equation}
where $J_{ij}[u]$ is a tensor valued operator, skew-symmetric with respect to the $L_2$ scalar product 
$\langle\ ,\rangle$, functionally dependent on $u$. Notice that with the choice
$J=\mathsf{J}_2$, and denoting $u_1=q$, $u_2=p$, (\ref{eq:brascp}) coincides with (\ref{eq:Pbra}).
On the other hand, the bracket (\ref{eq:brascp}) does not satisfy the Jacobi identity (hypothesis (iii) above), in general. We state without proof the following Proposition \cite{MaRa}, which characterises the Poisson brackets of the form (\ref{eq:brascp}). 
\begin{proposition}
\label{schout}
The bracket (\ref{eq:brascp}) satisfies the Jacobi identity, so that it is a Poisson bracket, if and only if
the skew-symmetric tensor $J[u]$ satisfies the Schouten identity
\begin{equation}
\label{eq:schout}
\sum_{s=1}^2\left(J_{is}D_{u_s} J_{jk}+J_{js}D_{u_s} J_{ki}+J_{ks}D_{u_s} J_{ij}\right)=0
\end{equation}
for all $u$ and all $i,j,k=1,2$. 
\end{proposition}
Here $D_{u_s}$ denotes the weak partial derivative with respect to $u_s$, defined in the usual way:
\begin{equation}
\label{eq:Dus}
\left(D_{u_s}f\right)h:=\frac{d}{d\epsilon}f[u_s+\epsilon h])\Big|_{\epsilon=0}\ ,
\end{equation}
for any $f$ functionally dependent on $u$. Observe that, for example, $D_{u_1}u_1=1$,
$D_{u_2}\partial_xu_2=\partial_x$, and so on.  Thus, any skew-symmetric tensor $J[u]$ satisfying the identity (\ref{eq:schout}) is a \emph{Poisson tensor}, i.e. it defines through (\ref{eq:brascp}) a Poisson bracket. An obvious but fundamental consequence of Proposition \ref{schout} is the following 
\begin{corollary}
Any skew-symmetric tensor $J$ independent of $u$ (i.e. constant on the phase space) is a Poisson tensor.
\end{corollary}
\begin{remark}
One does not require $J[u]$ to be non-degenerate, so that $J$ is allowed to have a nontrivial kernel. 
The functionals $F$ such that $J \nabla F = 0$ are called \emph{Casimir invariants} of the given Poisson structure, and represent constants of motion for all Hamiltonian systems: $\{H,F\}= 0$ 
for any $H \in \mathcal{A}(\Gamma)$.
\end{remark}
Within this framework, fixing a Hamiltonian $H[u]$ in the given Poisson algebra, the associated dynamics is defined in the usual way, namely 
\begin{equation}
\label{eq:genHam}
	u_t=\{u,H\}_J=J\nabla_u H\ ,
\end{equation}
to be read by components, $\nabla_u H$ being the functional gradient of $H[u]$. Of course, any functional $F$ evolves along the solutions of 
(\ref{eq:genHam}) according to $F_t=\{F,H\}_J$.
Hamiltonian dynamical systems, in the generalised Poisson sense, have the form (\ref{eq:genHam}), which includes the standard (symplectic) case. 

The fundamental feature of generalised Hamiltonian systems is their invariant character under any change of variables. 
\begin{proposition}
Any smooth change of variables $f:u\mapsto \tilde{u}=f[u]$ maps the Hamiltonian system 
$u_t=J\nabla_u H$ into the Hamiltonian system $\tilde{u}_t=\tilde{J}\nabla_{\tilde{u}} \tilde{H}$, where $\tilde{H}=H\circ f^{-1}$, whereas the transformed Poisson tensor $\tilde{J}$ is given by
\begin{equation}
\label{eq:transP}
\tilde{J}[\tilde{u}]:=(D_uf)J(D_uf)^T\Big|_{u=f^{-1}[\tilde{u}]}\ .
\end{equation}
The corresponding Poisson brackets are related, for any $F,G\in\mathcal{A}(\Gamma)$, by
\begin{equation}
\label{eq:Pbrarel}
\{F,G\}_J\circ f^{-1}=\{F\circ f^{-1},G\circ f^{-1}\}_{\tilde J}\ .
\end{equation}
\end{proposition}
In the latter formula, $D_u$ denotes the weak Jacobian of $u$, as defined in (\ref{eq:Dus}). The proof of the above Proposition is direct and not reported. The important point is the following: if $J$ is a Poisson tensor, its transformed $\tilde{J}$ under any $f$ is a Poisson tensor. Of course, the Hamilton equations are not
invariant in form under $f$, which happens if and only if $\tilde J=J$. \emph{Canonical transformations} are then defined as those transformations $f$ leaving the Poisson tensor invariant. In order to check the canonicity of a transformation $f$, it is easier to make use of (\ref{eq:Pbrarel}) which, with 
$J=\tilde J$, yields $\{F,G\}_J\circ f^{-1}=\{F\circ f^{-1},G\circ f^{-1}\}_J$.  
\begin{remark}
If $J=\mathsf{J}_2$, the transformation law (\ref{eq:transP}), together with the canonicity condition 
$\tilde J=J$, yields the requirement that the Jacobian $D_u f$ be symplectic. 
\end{remark}
The equation of motion (\ref{eq:genHam}) can be rewritten as $u_t=\mathcal{L}_Hu$, where 
the operator $\mathcal{L}_{H}\cdot=\{\cdot, G \}_J$, such that $\mathcal{L}_HF=\{F,H\}_J$ for any $F$, is 
the Lie derivative of $F$ in the direction of the Hamiltonian vector field $J\nabla H$. One can then formally solve the equation by exponentiation, which defines the flow $\Phi_H^t$ of the system, namely
\begin{equation}
	u(t)=e^{t\mathcal{L}_H}w:=\Phi^t_H(w)\ ,
\end{equation}
where $w=u(0)$ is an arbitrary initial condition. Of course the exponential operator above is defined, as usual, by its formal series
\begin{equation}
\label{eq:LieSeries}
 e^{t \mathcal{L}_H}= 1+ t\mathcal{L}_H+\frac{t^2}{2} \mathcal{L}_H^2+O(t^3) \, .
\end{equation}
Now, since the evolution equation $F_t=\{F,H\}_J=\mathcal{L}_HF$ of any functional $F$ is
solved by $e^{t\mathcal{L}_H}F(w)$, which must equal $F[u(t)]=F[\Phi^t_H(w)]$ for any initial condition $w$, one gets the useful relation 
\begin{equation}
\label{eq:exchange}
e^{t\mathcal{L}_H}F=F\circ \Phi^t_H\ ,
\end{equation}
which is known as the exchange Lemma; we will make use of it below. 

The Hamiltonian flow $\Phi^t_H:\Gamma\to\Gamma$ represents a one-parameter family of canonical transformations of $\Gamma$ into itself (the family is a group if the flow is global).
\begin{proposition}
For any $t$ such that $\Phi^t_H$ exists, and any pair of functionals $F$ and $G$, one has
\begin{equation}
\label{eq:flowcan}
	\{F,G\}_J \circ \Phi_H^t \;=\; \{F \circ \Phi_H^t,H\circ \Phi_H^t\}_J\ .
\end{equation}
\end{proposition}
\begin{proof}
Define $\Delta(t)$ the difference between the left and the right hand side of (\ref{eq:flowcan}), and observe that $\Delta(0)\equiv 0$. Making use of relation (\ref{eq:exchange}), and of the Jacobi identity, one gets
$d\Delta(t)/dt=\{\Delta(t),H\}_J=\mathcal{L}_H\Delta(t)$, whose solution is
$\Delta(t)=e^{t\mathcal{L}_H}\Delta(0)\equiv 0$. \qed
\end{proof}
\begin{remark}
In the above treatment, the Hamiltonian $H$ is arbitrary. It follows that any functional $G$, regarded as a Hamiltonian, generates a one-parameter family of canonical transformations, which is given by its flow 
$\Phi^s_G=e^{s\mathcal{L}_G}$, where $\mathcal{L}_G=\{\ ,G\}_J$. In the jargon, $G$ is called  the generating Hamiltonian, and $\mathcal{L}_G=d\Phi_G^s/ds|_{s=0}$ the generator of the transformation.
\end{remark}
As a final point of this section, we state a simple version of the N\"other theorem in the Poisson framework. 

\begin{theorem}
\label{thm:Nother}
If the Hamiltonian $H[u]$ is invariant with respect to the flow $e^{s\mathcal{L}_K}$ of generator 
$\mathcal{L}_K=\{\ ,K\}_J$, i.e. $e^{s\mathcal{L}_K}H=H$ for any $s$ close to zero, then $\{H,K\}_J=0$. 
\end{theorem}

\begin{proof}
The derivative of $e^{s\mathcal{L}_K}H=H$ with respect to $s$, at $s=0$, gives the result. \qed
\end{proof}
In the practice, one usually ``sees'' a certain symmetry of $H$, i.e. one is able to write down 
a certain transformation $\Psi^s$ such that $\Psi^0=1$ and $H\circ\Psi^s=H$ for any $s$ around zero.
Then, if $\Psi^s$ is a Hamiltonian flow, its generating Hamiltonian $K$ is a constant of motion of the given system.

\subsection{Perturbation theory}

The target of Hamiltonian perturbation theory, which goes back to Poincar\'e and Birkhoff, is the following. Given a Hamiltonian
\begin{equation}
\label{eq:Hpertu}
		H\;=\;H_0+\lambda H_1+\lambda^2 H_2 + O(\lambda^3)\ ,
\end{equation}
formally ordered with respect to the small parameter $\lambda$, one looks for a canonical transformation,
$\lambda$-close to the identity, erasing completely or in part the perturbation terms $H_{j\geq1}$ up to a given order (possibly infinite, as in the KAM theory). As is well known, the complete removal of the perturbation terms, even to the first few orders, is not possible, in general. The best one can do is instead to find a canonical transformation
setting $H$ in normal form, according to the following definition.
\begin{definition}
\label{def:nf}
The Hamiltonian $H_0+\lambda Z_1+\cdots+\lambda^nZ_n+O(\lambda^{n+1})$ is said to be in normal form to order $n\geq 1$ with respect to $H_0$ if $\mathcal{L}_{H_0}Z_j=\{Z_j,H_0\}=0$ for any $j=1,\dots,n$. 
\end{definition}
Observe that $Z_j\equiv0$ fits the normal form requirement, which means that the definition includes the possibility of complete removal of some perturbation terms.

The canonical transformation bringing the Hamiltonian (\ref{eq:Hpertu}) into normal form with respect to $H_0$, to order 
$\lambda^2$ included, is given by composing the flows of two unknown Hamiltonians $G_1$ and $G_2$, namely 
\begin{equation}
\label{eq:normtra}
		u\mapsto \tilde{u}= e^{-\lambda^2 \mathcal{L}_2} e^{-\lambda \mathcal{L}_1} u\ ,
\end{equation} 
where $\mathcal{L}_j:=\mathcal{L}_{G_j}$, $j=1,2$.  The \emph{inverse} transformation maps the Hamiltonian  (\ref{eq:Hpertu}) into
\begin{equation}
\label{eq:Htransf}
\begin{split}
\tilde{H}=\  &e^{\lambda^2 \mathcal{L}_2} e^{\lambda \mathcal{L}_1} H= H_0+
\lambda\left(\mathcal{L}_1 H_0+H_1\right) + \\
+ & \lambda^2\left( \mathcal{L}_2 H_0+\mathcal{L}_1 H_1 + \frac{1}{2} \mathcal{L}_1^2 H_0 + H_2 
\right)+O(\lambda^3)\ ,
\end{split}
\end{equation}
which is obtained by expanding the exponentials. The two generating Hamiltonians are then found by imposing that, according to the Definition \ref{def:nf}, the quantities 
\begin{equation}
\label{eq:Z1Z2}
	\begin{split}
		Z_1 \;&:=\; H_1+\mathcal{L}_1 H_0\ , \\
		Z_2 \;&:=\; \mathcal{L}_2 H_0+\mathcal{L}_1 H_1 + \frac{1}{2} \mathcal{L}_1^2 H_0 + H_2 
	\end{split}
\end{equation}
be first integrals of $H_0$. Observing that $\mathcal{L}_j H_0=-\mathcal{L}_{H_0}G_j$, the latter two equations for the four unknowns $Z_j$ and $G_j$, can be rewritten in the form
\begin{equation}
\label{eq:homo12}
	\begin{split}
		\mathcal{L}_{H_0}G_1 \;&:=\; H_1-Z_1\ , \\
		\mathcal{L}_{H_0}G_2 \;&:=\; \mathcal{L}_1 H_1 + \frac{1}{2} \mathcal{L}_1^2 H_0 + 
		H_2-Z_2\ . 
	\end{split}
\end{equation}
These equations have one and the same structure, namely
\begin{equation}
\label{eq:homo}
\mathcal{L}_{H_0}G_j=S_j-Z_j\ ,\ \ (j=1,2)
\end{equation}
with obvious definitions of the $S_j$. 
\begin{remark}
Looking for a transformation to an arbitrary order $n$, one finds at any order $j=1,\dots,n$ an equation of the form
(\ref{eq:homo}), where $S_j$ is a known quantity if all the equations up to order $j-1$ have been solved.
\end{remark}
Equation (\ref{eq:homo}) is known as the \emph{homological equation} of order 
$j$, which has to be solved determining the unknowns $Z_j$  and $G_j$ under the condition 
$\mathcal{L}_{H_0}Z_j=0$. 

In what follows we suppose that the flow $\Phi_{H_0}^s$ of $H_0$ is global (i.e. it exists for all $s\in\mathbb{R}$) and uniformly bounded with respect to $s$.
\begin{definition}
The time average of any $F$ along the unperturbed flow of $H_0$ is denoted by
\begin{equation}
\label{eq:tav}
	\langle F \rangle_0: =
	 \lim_{t\to\infty}\frac{1}{t} \int_0^t F\circ\Phi_{H_0}^s\ ds\ .
\end{equation}
If the flow of $H_0$ is $\tau$-periodic, i.e. $\Phi_{H_0}^\tau=1$, then
$\langle F\rangle _0=\frac{1}{\tau}\int_0^\tau F\circ\Phi_{H_0}^s ds$.
\end{definition}
\begin{lemma}
\begin{equation}
\label{eq:tavinv}
\mathcal{L}_{H_0}\langle F\rangle_0=0\ .
\end{equation}
\end{lemma}
\begin{proof}
Composing the left and right hand side of (\ref{eq:tav}) with the flow $\Phi_{H_0}^r$, one gets, on the right hand side, $\lim\frac{1}{t} \int_0^t F\circ\Phi_{H_0}^{s+r}\ ds=
\lim\frac{1}{t} \left(\int_r^0+\int_0^t+\int_t^{t+r}\right) F\circ\Phi_{H_0}^{a}\ da
=\lim\frac{1}{t} \int_0^t F\circ\Phi_{H_0}^{a}\ da$. Thus $\langle F\rangle_0\circ \Phi_{H_0}^r=F$, which implies (\ref{eq:tavinv}), and viceversa. \qed
\end{proof}

\begin{lemma}
\label{lem:homsol}
The solution of the homological equation (\ref{eq:homo}) is given by
\begin{equation}
\label{eq:ZjGj}
Z_j=\langle S_j\rangle_0\ \ ;\ \ G_j=\langle G_j\rangle_0+\lim_{t\to\infty}\frac{1}{t}
\int_0^t(s-t)e^{s\mathcal{L}_{H_0}}\left(S_j-\langle S_j\rangle_0\right)ds\ .
\end{equation} 
If the flow of $H_0$ is $\tau$-periodic, $G_j=\langle G_j\rangle_0+\frac{1}{\tau}\int_0^\tau s\ e^{s\mathcal{L}_{H_0}}\left(S_j-\langle S_j\rangle_0\right)ds$. 
\end{lemma}
\begin{proof}
Applying $e^{s\mathcal{L}_{H_0}}$ to equation (\ref{eq:homo}), taking into account the invariance of 
$Z_j$ (by the definition of normal form), and taking the time average, one gets the first of
(\ref{eq:ZjGj}) in the limit. By the latter result, the homological equation becomes 
$\mathcal{L}_{H_0}G_j=S_j-\langle S_j\rangle_0$. Applying $(s-t)e^{s\mathcal{L}_{H_0}}$ to the latter equation and time averaging, one gets the second of (\ref{eq:ZjGj}) in the limit. \qed
\end{proof}
\begin{remark}
The generating Hamiltonians $G_j$ solving the homological equation are defined up to their
average along the flow of $H_0$, i.e. up to an arbitrary constant of motion of $H_0$. Thus, both the normal form Hamiltonian and the transformation bringing to it are not unique. In the sequel, we make the
choice $\langle G_j\rangle_0\equiv0$. 
\end{remark}

\begin{theorem}[Averaging Principle]
\label{prop:PerturbationGeneral}
The canonical transformation 
\[u\mapsto \tilde u=
e^{-\lambda^2 \mathcal{L}_2} e^{-\lambda \mathcal{L}_1} u\ ,
\]
generated by
\begin{equation}
\label{eq:G1G2}
\begin{split}
G_1= & \lim_{t\to\infty}\frac{1}{t}\int_0^t(s-t)e^{s \mathcal{L}_0} \left(H_1-\langle H_1\rangle_0\right)ds\ ; \\
G_2= & \lim_{t\to\infty}\frac{1}{t}\int_0^t(s-t)e^{s \mathcal{L}_0} 
\left(S_2-\langle S_2\rangle_0\right)ds\ ; \\
S_2:= &H_2+\frac{1}{2}\left\{H_1,G_1\right\}+\frac{1}{2}\{\langle H_1\rangle_0,G_1\}\ ,
\end{split}
\end{equation}
maps  the perturbed Hamiltonian $H=H_0+\lambda H_1+\lambda^2 H_2 + O(\lambda^3)$ into the normal form $\tilde H=e^{\lambda^2 \mathcal{L}_2} e^{\lambda \mathcal{L}_1} H=H_0+\lambda Z_1+\lambda^2 Z_2+O(\lambda^3)$, explicitly given by
\begin{equation}
\label{eq:nfgen}
\tilde H=H_0+\lambda \langle H_1\rangle_0+ \lambda^2\left(
\langle H_2\rangle_0 +\frac{1}{2}\langle \{H_1,G_1\}\rangle_0\right)+O(\lambda^3)\ .
\end{equation}
\end{theorem}

\begin{proof}
By Lemma \ref{lem:homsol}, solving the first of the homological equations (\ref{eq:homo12}) yields
$Z_1$ and $G_1$. By substituting $\mathcal{L}_1H_0=Z_1-H_1=\langle H_1\rangle_0-H_1$
into the right hand side of the second of the homological equations  (\ref{eq:homo12}), one gets
the latter in the form $\mathcal{L}_{H_0}G_2=S_2-Z_2$, with $S_2$ as in (\ref{eq:G1G2}). Solving by 
Lemma \ref{lem:homsol} again yields $Z_2$ and $G_2$. \qed
\end{proof}
\begin{remark}
As a matter of fact, in order to get the normal form Hamiltonian (\ref{eq:nfgen}), one does not need to compute $G_2$. This is a general fact: $Z_{j+1}$ depends on $G_1,\dots,G_j$.
\end{remark}

\section{Hamiltonian field theory close to $q_{tt}=q_{xx}$}

We now come back to our problem, and solve it by applying all the tools introduced in the previous Section.

Let us start by considering a Hamiltonian $H=\oint\mathcal{H} dx$, whose density $\mathcal{H}$ does not depend explicitly on $t$ and $x$ and is an analytic function of $q_x$, $p$ and their spatial derivatives up to a certain finite order, in the neighbourhood of the origin. Since $\mathcal{H}$ is invariant under time, space and $q$ translations, Theorem \ref{thm:Nother} (N\"other) applies.

\begin{proposition}
$H=\oint\mathcal{H} dx$, $I=\oint q_xp\ dx$ and $P=\oint p\ dx$ are the three first integrals corresponding to the symmetries $t\to t+s$, $x\to x+s$ and $q\to q+s$, respectively. Moreover,
$\{I,P\}=0$, so that the three first integrals are in involution.
\end{proposition} 

\begin{proof}
The conservation of $H$ is obvious. The Hamilton equations for $I$ at time $s$ are: $q_s=q_x$ and $p_s=p_x$,  whose solution is $q(t,x+s)$ and $p(t,x+s)$, clearly corresponding to the $x$-translation. 
The Hamilton equations for $P$ are $q_s=1$, $p_s=0$, solved by $q(t,x)+s$ and $p(t,x)$, corresponding to the $q$-translation. Finally, observe that
$\{I,P\}_{q,p}=\oint (\delta I/\delta q) (\delta P/\delta p)dx = -\oint p_x\ dx = 0$. \qed
\end{proof}

\begin{remark}
\label{P0}
One can always restrict the dynamics to the submanifold $P=\oint p\ dx=0$ by the canonical transformation $q=q'$, $p=P+p'$. 
\end{remark}

For the sake of convenience, we repeat below the definition of the class of Hamiltonian functionals considered, with the appropriate grading.

\begin{definition}\label{def:DefGrading}
The perturbative ordering of the Hamiltonian $H$ is defined by the following scaling:
\begin{equation}
H_\lambda:=\frac{1}{\lambda^4}\oint \mathcal{H}(\lambda^2q_{x},\lambda^2p,\lambda^3q_{xx},\lambda^3p_x,\dots,\lambda^6q_{xxxxx},\lambda^6p_{xxxx})\ dx\ .
\end{equation}
\end{definition} 

By Taylor expanding in powers of $\lambda$, close to $\lambda=0$, assuming without loss of generality that  
$\mathcal{H}|_{(q,p)=0}=0$, and taking into account Remark \ref{P0},
one gets
\begin{equation}\label{eq:HamiltonianGeneric}
H_\lambda=H_0+ \lambda H_1 +\lambda^2 H_2+\lambda^3 H_3 +\lambda^4 H_4 +\cdots\ , 
\end{equation}
where 
\begin{equation}
\label{eq:H0}
H_0=\oint \frac{a p^2+ b (q_x)^2}{2} dx +c I 
\end{equation}
with $a$, $b$ and $c$ some constants and $I=\oint q_xp\ dx$;
\begin{equation}
H_1=\oint d_1 q_xp_x\ dx\ ;
\end{equation}
\begin{eqnarray}
H_2=\oint&& \left[ e_1 (q_x)^3 +e_2 p^3 +e_3 (q_x)^2p +e_4 q_xp^2 + e_5(q_{xx})^2 +\right.\nonumber\\
+&& \left. e_6 (p_x)^2+e_7 q_{xx}p_x\right]\ dx\ ;
\end{eqnarray}
\begin{equation}
H_3=\oint\left[f_1 (q_x)^2p_x+ f_2 q_{xx}p^2 +f_3 q_{xx}p_{xx} \right]\ dx \, ;
\end{equation}
\begin{eqnarray}\label{eq:H4General}
H_4=\oint&&\left[g_1(q_x)^4+g_2p^4+g_3 (q_x)^2p^2+g_4(q_x)^3p+g_5q_xp^3\right. +\nonumber\\
+&&\left. g_6(q_{xx})^2q_x + g_7(q_{xx})^2p+g_8(p_x)^2q_x+g_9(p_x)^2p \right. +\nonumber\\
+&&\left. g_{10} q_{xxx} p^2+g_{11}(q_x)^2p_{xx}+g_{12}(q_{xxx})^2+g_{13}(p_{xx})^2 \right. +\nonumber\\
+&& \left. g_{14}q_{xxxx}p_x\right]\ dx\ ,
\end{eqnarray}
and so on. Here $d_1,e_1,\dots,g_{14}$ are given constants. 

\begin{remark}
Since $\mathcal{H}$ is independent of $x$, the density of each $H_j$ is independent of $x$. It follows that 
$\{I,H_j\}=0$ for any $j\geq0$.
\end{remark}

\begin{proposition}
\label{condab}
If the constants $a:=\partial^2\mathcal{H}/\partial p^2|_0$ and 
$b:=\partial^2\mathcal{H}/\partial(q_x)^2|_0$ appearing in (\ref{eq:H0}) are different from zero and have the same sign, there exists a time-dependent canonical transformation which brings the Hamiltonian $H_0$ in the canonical wave equation form $K_0=\frac{1}{2}\oint [p^2+(q_x)^2]dx$, and preserves the structure of the perturbations $H_j$ to any order $j\geq0$.
\end{proposition}

\begin{proof}
Let $a=\sigma|a|$ and $b=\sigma|b|$, with $\sigma=\pm1$. One first performs the canonical rescaling 
$q=\sqrt{|a|}\ q'$, $p=\sqrt{|b|}\ p'$, $H=\sigma|ab| H'$, $t=\sigma t'$, which brings $H_0$ into $K_0+c'I$, where 
$c'=\sigma c/\sqrt{|ab|}$. Then one performs the transformation $(q',p')=\Phi_{c'I}^t(q'',p'')=\Phi_I^{c't}(q'',p'')$, where
$\Phi_I^t$ denotes the flow of $I=\oint q_xp\ dx$. The latter transformation is canonical and erases $c'I$. 
Clearly, both transformations do not change the structure of any $H_j$ nor the value of the coefficients of the Hamiltonians $H_1,\dots,H_4$. Observe that the flow of $I$ is the left translation of $(q,p)$, so that it is global and preserves the regularity of the initial condition. \qed
\end{proof}

\begin{remark}
Consider $K_0+\lambda H_1=\frac{1}{2}\oint[p^2+(q_x)^2+2\lambda d_1q_xp_x]dx$. Its Hamilton equations read
\[
q_t=p-\lambda d_1q_{xx}\ \ ;\ \ p_t=q_{xx}+\lambda d_1p_{xx}\ .
\]
Both $q$ and $p$ satisfy the linear Boussinesq equation
\[
u_{tt}=u_{xx}+(\lambda d_1)^2u_{xxxx}\ .
\]
\end{remark}

The condition on $a$ and $b$ in the Proposition~\ref{condab} above identifies the elliptic fixed points in the given class of Hamiltonians. One is then left with the problem of simplifying the dynamics of $K_0+\lambda H_1+\lambda^2H_2+\cdots$. The perturbations to various order have the structure listed above and no further simplification can be made, in general. However, there is a relevant class of Hamiltonians that display a much simpler structure, namely the class of \emph{mechanical Hamiltonians} of the form $\mathcal{H}=p^2/2+\mathcal{U}$, where $\mathcal{U}$ depends only on $q_x$ and its derivatives.
Such Hamiltonians usually arise as the continuum limit of some lattice system, the notable case being just that of the vibrating string. 

\begin{proposition}\label{prop:MechanicalHamiltonianGEN}
Suppose that $\mathcal{H}=p^2/2+\mathcal{U}(q_x,q_{xx},\dots,q_{xxxxx})$. Then, if the condition
$b:=\partial^2\mathcal{U}/\partial(q_x)^2|_0>0$ holds, $H_0$ can be brought in the canonical wave form $K_0$, $H_1=H_3\equiv0$, and
\[
H_2=\oint \left[\alpha_1 (q_x)^3 +\alpha_2(q_{xx})^2\right]\ dx\ ;
\]
\[
H_4=\oint\left[\beta_1(q_x)^4+\beta_2(q_{xx})^2q_x+\beta_3(q_{xxx})^2\right]\ dx\ .
\]
\end{proposition}
\begin{proof}
The momentum $p$ cannot appear out of $H_0$, by definition. Notice that in this case there is no term proportional to
$I$ in $H_0$. \qed
\end{proof}
In the latter significant case one can obviously rename $H_2\to H_1$ and $H_4\to H_2$, $\lambda^2\to\lambda$.

\subsection{Traveling waves}

The equations of motion associated to $K_0=\oint \frac{p^2+(q_x^2)}{2}\ dx$ reduce to the wave-equation for the field $q$:
\begin{equation}
	q_t=p\ \ ;\ \ p_t=q_{xx}\ , \qquad \Longleftrightarrow \qquad  q_{tt}=q_{xx} \, .
\end{equation}
In order to simplify the analysis of perturbations of the wave equation, it is convenient to perform a change of variables that maps the functions $(q,p)$ into the \emph{Riemann invariants} $(u,v)$:
\begin{equation}
\label{eq:WavesRiemannInv}
		u \;=\; \dfrac{q_x+p}{\sqrt{2}}\ \ ;\ \ v \;=\; \dfrac{q_x-p}{\sqrt{2}}\ .
\end{equation}
The equations of motion for $u$ and $v$ are the left and right translation equation, respectively:
\begin{equation}
	\begin{cases}
		u_t=u_x \\
		v_t=-v_x
	\end{cases} \, .
\end{equation}
Indeed, the solution of the above system corresponding to the initial condition
$(u_0(x),v_0(x))$ is $(u_0(x+t),v_0(x-t))$, i.e. a rigid translation of the initial profiles. The flow of the wave equation, that is used to compute normal forms, is particularly manageable in these new variables, being a left translation for $u$ and a right translation for $v$ (at positive times).

The change of variables (\ref{eq:WavesRiemannInv}) is not canonical and it maps the standard Poisson tensor $\mathsf{J}_2$ into the Gardner tensor
\begin{equation}
\label{eq:Gard}
	J\;=\; \begin{pmatrix}
		\partial_x & 0 \\
		0 & -\partial_x
	\end{pmatrix} \, .
\end{equation}
In particular, as can be checked, formula (\ref{eq:transP}) for the transformation 
(\ref{eq:WavesRiemannInv}) reads
\[
D_{q,p}(u,v)\begin{pmatrix} 0 & 1 \\ -1 & 0 \end{pmatrix} D_{q,p}^T(u,v)=
\begin{pmatrix} \partial_x & 0 \\ 0 & -\partial_x \end{pmatrix}\ .
\]
The Hamiltonian $K_0$, expressed in terms of $(u,v)$ reads $K_0=\oint\frac{u^2+v^2}{2}\ dx$, so that the translation equations for $u$ and $v$ are the Hamilton equations associated to $K_0$ in the Gardner structure.

The explicit expression of the Hamiltonians \eqref{eq:H0}-\eqref{eq:H4General} in the $(u,v)$ variables is:
\begin{equation}
	K_0\;=\;\oint \frac{u^2+v^2}{2} \, dx \, ;
\end{equation}
\begin{equation}
	H_1\;=\; \oint \frac{d_1}{\sqrt{|ab|}} u v_x \, dx \, ;
\end{equation}
\begin{equation}
	\begin{split}
		H_2 \;=\; \oint &\Big\{ \frac{1}{2^{3/2}} \Big[ \big({\textstyle \frac{e_1}{|b|^{3/2}}+\frac{e_2}{|a|^{3/2}}+\frac{e_3}{|b| \sqrt{|a|}}+\frac{e_4}{|a| \sqrt{|b|}}}\big)u^3 \\
		&+\big({\textstyle \frac{e_1}{|b|^{3/2}}-\frac{e_2}{|a|^{3/2}}-\frac{e_3}{|b| \sqrt{|a|}}+\frac{e_4}{|a| \sqrt{|b|}}}\big)v^3\\
		&+\big({\textstyle \frac{3 e_1}{|b|^{3/2} }-\frac{3 e_2}{|a|^{3/2}}+\frac{e_3}{|b|\sqrt{|a|}}-\frac{e_4}{|a|\sqrt{|b|}}} \big)u^2 v \\
		&+\big({\textstyle \frac{3 e_1}{|b|^{3/2} }+\frac{3 e_2}{|a|^{3/2}}-\frac{e_3}{|b|\sqrt{|a|}}-\frac{e_4}{|a|\sqrt{|b|}}} \big)u v^2 \Big]+ \\
		&+\frac{1}{2} \Big[ \big({\textstyle \frac{e_5}{|b|} + \frac{e_6}{|a|} + \frac{e_7}{\sqrt{|ab|}}} \big)u_x^2 +\\
		&+\big({\textstyle \frac{e_5}{|b|} + \frac{e_6}{|a|} - \frac{e_7}{\sqrt{|ab|}}} \big)v_x^2 \Big] \Big\} \, dx \, .
	\end{split}
\end{equation}

\subsection{The generic case}

In order to perform a canonical transformation as stated in Proposition \ref{prop:PerturbationGeneral}, one has to compute time averages, as required in Theorem \ref{prop:PerturbationGeneral}. General formulas applying to the case of an unperturbed flow consisting of left/right translations are provided in the next lemma.

%\begin{svgraybox}
\begin{lemma}\label{lem:PropertiesIntegrals} Suppose that $f$ and $g$ are continuous functions on $\mathbb{T}$. Then 
\begin{eqnarray}
\oint \oint f(x\pm s) \, dx \, ds &=&\oint f(x) \, dx  \ ; \label{eq:IntegralU}\\
	\oint \oint f(x\pm s) g(x \mp s) \, dx \, ds &=&\oint f(x) \, dx \oint g(y) \, dy \ ;\label{eq:IntegralUV}\\
	\int_0^1 \oint s \, f(x\pm s) g(x \mp s) \, dx \, ds &=& \nonumber \\
	&\,&\!\!\!\!\!\!\!\!\!\!\!\!\!\!\!\!\!\!\!\!\!\!\!\!\!\!\!\!\!\!\!\!\!\!\!\!\!\!\!\!\!=\frac{1}{2} \oint f(x) \, dx \oint g(y) \, dy \pm \frac{1}{2} \oint g(x) \, \partial_x^{-1} f(x) \, dx \ , \label{eq:IntegralSUV}
\end{eqnarray}
where $\partial_x^{-1}f(x)$ denotes the unique primitive of $f$ with zero average on $\mathbb{T}$.
\end{lemma}
%\end{svgraybox}

\begin{proof}
All these proofs consist of straightforward computations in Fourier space. First, we prove  \eqref{eq:IntegralUV}:
\[
	\begin{split}
		\oint \oint f(x\pm s) g(x \mp s) \, dx \, ds \;&=\; \int_0^1 \int_0^1 \sum_{k,k' \in \mathbb{Z}} \hat{f}_k \hat{g}_{k'} e^{2 \pi \imath k(x \pm s)} e^{2 \pi \imath k' (x \mp s)} \, dx \, ds \\
		&=\;\sum_{k,k' \in \mathbb{Z}} \hat{f}_k \hat{g}_{k'} \delta_{k+k',0} \delta_{k-k',0} \;=\; \hat{f}_0 \hat{g}_0 \, .
	\end{split}
\]
From here, \eqref{eq:IntegralU} follows by choosing $g=1$. In order to prove \eqref{eq:IntegralSUV}, we Fourier-transform the LHS:
\[
	\int_0^1 \oint s \, f(x\pm s) g(x \mp s) \, dx \, ds \;=\; \sum_{k \in \mathbb{Z}} \hat{f}_k \hat{g}_{-k} \int_0^1 s e^{\pm 4 \pi \imath k s } \, ds\ .
\]
It remains to notice that
\[
	\int_0^1 s e^{\pm 4 \pi \imath k s } \, ds \;=\; \delta_{k,0} \int_0^1 s \, ds + (1-\delta_{k,0}) \int_0^1 s e^{\pm 4 \pi \imath k s } \, ds = \frac{1}{2} \delta_{k,0} \pm \frac{1}{2} \frac{1}{2 \pi \imath k}(1-\delta_{k,0})
\]
and to recognise that $1/(2\pi \imath k)$ is the Fourier-multiplier corresponding to the operator $\partial_x^{-1}$. \hfill \qed
\end{proof}

%\begin{svgraybox}
\begin{proposition}
	There exists a (formal) near-to-identity, canonical transformation $(u,v) \mapsto (\hu,\hv)$ mapping $H_\lambda$ into	
	\begin{equation}
		\tilde{H}_\lambda \;=\; K_0 + \lambda^2 Z_2+O(\lambda^3)\ ,
	\end{equation}
	where
	\begin{equation}
		K_0 \;=\; \oint \frac{\hu^2+\hv^2}{2} \, dx \, ;
	\end{equation}
	\begin{equation}
		\begin{split}
			Z_2 \;=\; \oint &\Big\{ \frac{1}{2^{3/2}} \Big[ \big({\textstyle \frac{e_1}{|b|^{3/2}}+\frac{e_2}{|a|^{3/2}}+\frac{e_3}{|b| \sqrt{|a|}}+\frac{e_4}{|a| \sqrt{|b|}}}\big)\hu^3 + \\
		&+\big({\textstyle \frac{e_1}{|b|^{3/2}}-\frac{e_2}{|a|^{3/2}}-\frac{e_3}{|b| \sqrt{|a|}}+\frac{e_4}{|a| \sqrt{|b|}}}\big)\hv^3 \Big]+ \\
		&+\frac{1}{2} \Big[ \big({\textstyle \frac{e_5}{|b|} + \frac{e_6}{|a|} + \frac{e_7}{\sqrt{|ab|}}-\frac{d_1^2}{2|ab|}} \big)\hu_x^2 +\\
		&+\big({\textstyle \frac{e_5}{|b|} + \frac{e_6}{|a|} - \frac{e_7}{\sqrt{|ab|}}-\frac{d_1^2}{2|ab|}} \big)\hv_x^2 \Big] \Big\} \, dx \, .
		\end{split}
	\end{equation}
\end{proposition}
%\end{svgraybox}

\begin{proof} 

\noindent \emph{First perturbative step}: Using \eqref{eq:nfgen} and \eqref{eq:IntegralUV} one has $Z_1=0$:
\[
	\begin{split}
		Z_1 \;&=\; \int_0^1 e^{s \mathcal{L}_{H_0}} H_1 \, ds \\
			&=\; \int_0^1 \oint \frac{d_1}{\sqrt{|ab|}} u(x+s) v_x(x-s) \, dx \, ds \\
			&\overset{\text{\eqref{eq:IntegralUV}}}{=}\; \frac{d_1}{\sqrt{|ab|}}\oint u(x) \, dx \oint v_y(y) \, dy \;=\; 0 
	\end{split}
\]
where in the last step we used that the $v_y$ has zero-average.

\noindent \emph{Additional term at second order}: We need the expression of $G_1$ to compute $Z_2$. Using \eqref{eq:nfgen} and \eqref{eq:IntegralSUV} we have
\[
	\begin{split}
		G_1\;&=\; \int_0^1 s e^{s \mathcal{L}_{H_0}} H_1 \, ds \\
			&=\;\frac{d_1}{\sqrt{|ab|}}  \int_0^1 \oint s \, u(x+s) v_x(x-s) \, dx \, ds \\
			&\overset{\text{\eqref{eq:IntegralSUV}}}{=}\; -\frac{d_1}{2\sqrt{|ab|}} \oint u v \, dx \, .
	\end{split}
\]
The computation of functional derivatives yields:
\[
	\frac{\delta G_1}{\delta u} \;=\; - \frac{d_1}{2 \sqrt{|ab|}} v \, ; \qquad \frac{\delta G_1}{\delta v} \;=\; -\frac{d_1}{2 \sqrt{|ab|}} u \, ;
\]
\[
	\frac{\delta H_1}{\delta u} \;=\; \frac{d_1}{\sqrt{|ab|}} v_x \, ; \qquad \frac{\delta H_1}{\delta v} \;=\; -\frac{d_1}{\sqrt{|ab|}} u_x\ ,
\]
and one finally obtains
\[
	\begin{split}
		\{H_1,G_1\} \;&=\; \oint \left(\frac{\delta H_1}{\delta u} \partial_x \frac{\delta G_1}{\delta u} - \frac{\delta H_1}{\delta v} \partial_x \frac{\delta G_1}{\delta v} \right) \, dx \\
		&=\;-\frac{d_1^2}{2 |ab|} \oint \left(v_x^2+u_x^2 \right) \, dx \, .
	\end{split}
\]

\noindent \emph{Computation of the second order normal form:} Using \eqref{eq:nfgen}, one has to time-average (with respect to the unperturbed flow of $K_0$) the following expression:
\[
	\begin{split}
		H_2+\frac{1}{2} &\{H_1-Z_1,G_1\} \;=\; \oint \Big\{ \frac{1}{2^{3/2}} \Big[ \big({\textstyle \frac{e_1}{|b|^{3/2}}+\frac{e_2}{|a|^{3/2}}+\frac{e_3}{|b| \sqrt{|a|}}+\frac{e_4}{|a| \sqrt{|b|}}}\big)u^3 \\
		&+\big({\textstyle \frac{e_1}{|b|^{3/2}}-\frac{e_2}{|a|^{3/2}}-\frac{e_3}{|b| \sqrt{|a|}}+\frac{e_4}{|a| \sqrt{|b|}}}\big)v^3+\big({\textstyle \frac{3 e_1}{|b|^{3/2} }-\frac{3 e_2}{|a|^{3/2}}+\frac{e_3}{|b|\sqrt{|a|}}-\frac{e_4}{|a|\sqrt{|b|}}} \big)u^2 v \\
		&+\big({\textstyle \frac{3 e_1}{|b|^{3/2} }+\frac{3 e_2}{|a|^{3/2}}-\frac{e_3}{|b|\sqrt{|a|}}-\frac{e_4}{|a|\sqrt{|b|}}} \big)u v^2 \Big]+ \\
		&+\frac{1}{2} \Big[ \big({\textstyle \frac{e_5}{|b|} + \frac{e_6}{|a|} + \frac{e_7}{\sqrt{|ab|}}-\frac{d_1^2}{2 |ab|}} \big)u_x^2 +\big({\textstyle \frac{e_5}{|b|} + \frac{e_6}{|a|} - \frac{e_7}{\sqrt{|ab|}}-\frac{d_1^2}{2 |ab|}} \big)v_x^2 \Big] \Big\} \, dx\ .
	\end{split}
\]
As a consequence of \eqref{eq:IntegralUV} and under the assumption of $\oint u \, dx = \oint v \, dx =0$:
\[
	\int_0^1 \oint u^2(x+s) v(x-s) \, dx \, ds \;=\; \left(\oint u^2(x) \, dx \right) \left( \oint v(x) \, dx \right) \;=\; 0\ ;
\]
\[
	\int_0^1 \oint u(x+s) v^2(x-s) \, dx \, ds \;=\; \left(\oint u(x) \, dx \right) \left( \oint v^2(x) \, dx \right) \;=\; 0\ .
\]
Moreover
\[
	\int_0^1 \oint u^3(x+s) \, dx \, ds \;=\; \oint u^3(x) \, dx\ ;
\]
\[
	\int_0^1 \oint v^3(x+s) \, dx \, ds \;=\; \oint v^3(x) \, dx\ ;
\]
\[
	\int_0^1 \oint u_x^2 (x+s) \, dx \, ds \;=\; \oint u_x^2(x) \, dx\ ;
\]
\[
	\int_0^1 \oint v_x^2 (x+s) \, dx \, ds \;=\; \oint v_x^2(x) \, dx\ ,
\]
and this completes the proof. \hfill \qed
\end{proof}

\begin{remark}
$\tilde{H}_\lambda$ is \emph{always} the Hamiltonian of a pair of counter-propagating Korteweg-de Vries equations (up to a small remainder), i.e. its vector field $J\nabla\tilde H_\lambda$ is of the form (\ref{eq:ord2}). Such a result is somehow expected from, and in agreement with the existing results treating particular cases in the literature, among which those concerning the FPU problem (starting with the seminal work of Zabusky and Kruskal \cite{ZK}) and the propagation of surface water waves (where the first deduction of the KdV equation goes back to Boussinesq
\cite{Bous}). 
\end{remark}

\subsection{The mechanical case}\label{subsec:MechanicalHamiltonians}

For mechanical Hamiltonians of the form $\mathcal{H}=p^2/2+\mathcal{U}$, where $\mathcal{U}$ depends on $q_x$ and its derivatives, starting from Proposition \ref{prop:MechanicalHamiltonianGEN} and repeating the analysis made in the general case, we perform the change of variables $(q,p) \mapsto (u,v)$, which yields
\begin{equation}
	K_0\;=\; \oint\frac{u^2+v^2}{2} \, dx \, ,
\end{equation}
\begin{equation}
	H_2\;=\; \oint\left[\frac{\alpha_1}{2^{3/2}} \left(u^3+3u^2v+3 u v^2 + v^3 \right)+\frac{\alpha_2}{2} \left((u_x)^2+2u_xv_x+(v_x)^2 \right)\right] \, dx \, ,
\end{equation}
\begin{equation}
	\begin{split}
		H_4\;&=\;\oint \Big\{\beta_1 \left[\frac{u^4+4u^3v+6u^2v^2+4u v^3+v^4}{4}\right]\\
		&\quad+\beta_2 \left[\frac{(u_x)^2+2u_x v_x +(v_x)^2}{2} \right]\frac{u+v}{\sqrt{2}}\\
		&\quad +\frac{\beta_3}{2} [(u_{xx})^2+2u_{xx} v_{xx}+(v_{xx})^2] \Big\} \, dx \, .
	\end{split}
\end{equation}

%\begin{svgraybox}
\begin{proposition}\label{prop:GenericNormalForm}
	There exists a (formal) near-to-identity, canonical transformation $(u,v) \mapsto (\hu,\hv)$ mapping $H_\lambda$ into	
	\begin{equation}
		\tilde{H}_\lambda \;=\; K_0 + \lambda^2 Z_2+\lambda^4 Z_4 +O(\lambda^6)\ ,
	\end{equation}
	where
	\begin{equation}
		K_0 \;=\; \oint \frac{\hu^2+\hv^2}{2} \, dx \, ,
	\end{equation}
	\begin{equation}
			Z_2 \;=\; \oint \left[ \frac{\alpha_1}{2^{3/2}}(\tilde{u}^3+\tilde{v}^3)+\frac{\alpha_2}{2} (\tilde{u}_x^2+\tilde{v}_x^2) \right] \, dx \, ,
	\end{equation}
	\begin{equation}
		\begin{split}
			Z_4 \;=\; &\oint \Big\{ \Big(\frac{\beta_1}{4}-\frac{9\alpha_1^2}{16} \Big)(\tilde{u}^4+\tilde{v}^4)+\Big(\frac{\beta_2}{2^{3/2}}-\frac{3\alpha_1\alpha_2}{\sqrt{2}} \Big)\big[\tilde{u}(\tilde{u}_x)^2+\tilde{v} (\tilde{v}_x)^2 \big]+ \\
			& +\Big(\frac{\beta_3}{2}-\frac{\alpha_2^2}{2} \Big) \big[ (\tilde{u}_{xx})^2+(\tilde{v}_{xx})^2 \big] \Big\} \, dx +\Big(\frac{3 \beta_1}{2}-\frac{9 \alpha_1^2}{2} \Big) \langle \tilde{u}^2 \rangle \langle \tilde{v}^2 \rangle+\\
			&+\frac{9 \alpha_1^2}{16} \big( \langle \tilde{u}^2 \rangle^2+\langle \tilde{v}^2 \rangle^2 \big) \, .
		\end{split}
	\end{equation}
\end{proposition}
%\end{svgraybox}

\begin{proof}
	\emph{First perturbative step:} using Proposition \ref{prop:PerturbationGeneral} we have
	\[
		\begin{split}
			Z_2\;&=\; \int_0^1 e^{s \mathcal{L}_{H_0}} H_2 \, ds \\
		&\overset{\text{\eqref{eq:IntegralUV}}}{=}\;\oint \Big\{ \frac{\alpha_1}{2^{3/2}}(u^3+v^3)+\frac{\alpha_2}{2}[(u_x)^2+(v_x)^2] \Big\} \, dx \\
		&+\frac{3 \alpha_1}{2^{3/2}} \big(\langle u^2 \rangle \langle v \rangle + \langle u \rangle \langle v^2 \rangle \big)\ ;
		\end{split}
	\]
	here the last term vanishes because $\langle u\rangle=\oint u \, dx =0$ and $\langle v\rangle=\oint v \, dx = 0$.
	
	\noindent \emph{Generator of the the first order transformation}: 
	\[
		\begin{split}
			G_2 \;&=\; \int_0^1 s e^{s \mathcal{L}_{H_0}} (H_2-Z_2) \, ds \\
				&=\int_0^1 \oint s \Big\{\frac{3 \alpha_1}{2^{3/2}} \big[u^2(x+s) v(x-s)+ u(x+s) v^2(x-s)\big]\\
				&\qquad +\alpha_2 u_x(x+s)v_x(x-s) \Big\} \, dx \, ds\\
				&\overset{\text{\eqref{eq:IntegralSUV}}}{=}\frac{3 \alpha_1}{2^{5/3}} \big( \langle u^2 \rangle \langle v \rangle + \langle u \rangle \langle v^2 \rangle \big)+\frac{3 \alpha_1}{2^{5/2}} \Big( \oint v^2 \partial_x^{-1} u \, dx + \oint v \partial_x^{-1} u^2 \, dx \Big)\\
				&\qquad + \frac{\alpha_2}{2} \oint u v_x \, dx \\
				&=\;\frac{3 \alpha_1}{2^{5/2}} \Big( \oint v^2 \partial_x^{-1} u \, dx + \oint v \partial_x^{-1} u^2 \, dx \Big) + \frac{\alpha_2}{2} \oint u v_x \, dx\ ,
		\end{split}
	\]
	where in the last step we used $\langle u \rangle = 0$ and $\langle v \rangle=0$.
	Making use of the functional derivatives
	\[
		\frac{\delta G_1}{\delta u} \;=\; \frac{3 \alpha_1}{2^{3/2}} \Big[-u \partial_x^{-1} v - \frac{1}{2} \partial_x^{-1} v^2 \Big] + \frac{\alpha_2}{2} v_x \, ;
	\]
	\[
		\frac{\delta G_1}{\delta v} \;=\; \frac{3 \alpha_1}{2^{3/2}} \Big[\frac{1}{2}\partial_x^{-1} u^2+v \partial_x^{-1} u \Big] - \frac{\alpha_2}{2}u_x\, ;
	\]
	\[
		\frac{\delta (H_2-Z_2)}{\delta u} \;=\; \frac{3 \alpha_1}{2^{3/2}} \big[ 2u v + v^2\big]-\alpha_2 v_{xx}\, ;
	\]
	\[
			\frac{\delta (H_2-Z_2)}{\delta v} \;=\; \frac{3 \alpha_1}{2^{3/2}} \big[ u^2+ 2u v\big]-\alpha_2 u_{xx}\, ,
	\]
	we can compute the Poisson bracket
	\[
		\begin{split}
			\{ H_2-Z_2,G_2\} \;&=\; \oint \Big[\frac{\delta (H_2-Z_2)}{\delta u} \partial_x \frac{\delta G_2}{\delta u} - \frac{\delta (H_2-Z_2)}{\delta v} \partial_x \frac{\delta G_2}{\delta v} \Big] \, dx \, .
		\end{split}
	\]
	Since we don't need its full expression, we can use \eqref{eq:IntegralUV} to simplify computations and consider only those terms that don't vanish after taking the average with respect to the flow of $K_0$. We obtain
	\[
		\begin{split}
			\Big\langle \{H_2-Z_2,G_2\} \Big \rangle_{0} \;&=\;\oint \Big[-\frac{9 \alpha_1^2}{16} \big( u^4+v^4 \big)-\frac{3\alpha_1\alpha_2}{\sqrt{2}} \big(u_x^2 u+v_x^2 v\big)+ \\
			&-\frac{\alpha_2^2}{2}\big( (u_{xx})^2+(v_{xx})^2 \big) \Big] \, dx+\frac{9 \alpha_1^2}{16} \big( \langle u^2 \rangle^2+\langle v^2 \rangle^2 \big)\\ & -\frac{9 \alpha_1^2}{2} \langle u^2 \rangle \langle v^2 \rangle\ ,
		\end{split}
	\]
	whereas
	\[
		\begin{split}
			\langle H_4 \rangle_{0} \;&=\; \oint \Big\{ \frac{\beta_1}{4}\big(u^4+v^4\big)+\frac{\beta_2}{2^{3/2}} \big[ u(u_x)^2+v(v_x)^2 \big]+\frac{\beta_3}{2} \big[ (u_{xx})^2+(v_{xx})^2 \big]\Big\} \, dx \\
			&+\frac{3 \beta_1}{2} \langle u^2 \rangle \langle v^2 \rangle \, .
		\end{split}
	\]
Summing the right hand sides of the two previous equations we get
	\[
		\begin{split}
			Z_4 \;&=\; \oint \Big\{ \Big(\frac{\beta_1}{4}-\frac{9\alpha_1^2}{16} \Big)(u^4+v^4)+\Big(\frac{\beta_2}{2^{3/2}}-\frac{3\alpha_1\alpha_2}{\sqrt{2}} \Big)\big[u(u_x)^2+v (v_x)^2 \big]+ \\
			&\qquad +\Big(\frac{\beta_3}{2}-\frac{\alpha_2^2}{2} \Big) \big[ (u_{xx})^2+(v_{xx})^2 \big] \Big\} \, dx +\Big(\frac{3 \beta_1}{2}-\frac{9 \alpha_1^2}{2} \Big) \langle u^2 \rangle \langle v^2 \rangle\\& \qquad+\frac{9 \alpha_1^2}{16} \big( \langle u^2 \rangle^2+\langle v^2 \rangle^2 \big) \, .
		\end{split}
	\]
	\noindent
	\qed
\end{proof}
Here, as in the generic case, $Z_2$ is in the KdV hierarchy, i.e. the vector field $J\nabla(\tilde K_0+\lambda^2 Z_2)$ has the form of the right-hand side of 
(\ref{eq:ord2}). On the other hand, $Z_4$ is not, in general, in the KdV hierarchy: the two components of its vector field $J\nabla Z_4$ are not proportional to $\kappa_5$ (as defined in (\ref{eq:kappa5})), which is due to the impossibility to fit all the required constraints on its parameters, in general. However, it is still possible to get a dynamics within the KdV hierarchy to order 
$\lambda^4$ by applying the Kodama normalisation procedure to the vector field $J\nabla(\tilde K_0+\lambda^2 Z_2+\lambda^4 Z_4)$. Although such a normalisation is noncanonical, in principle, it actually yields a system of equations in the form (\ref{eq:ord4}). Neglecting the remainder, these equations turn out to be Hamiltonian a fortiori, with the correct Gardner-Poisson tensor (\ref{eq:Gard}).
The deep reason behind this fact is far from being deeply understood, at present.

Concrete examples are discussed in the next Section \ref{sec:Applications}, where we also provide an explicit example of Kodama transformation. 

\section{Applications}
\label{sec:Applications}

\subsection{The Fermi-Pasta-Ulam problem}\label{subsec:FPU}
The Fermi-Pasta-Ulam (FPU) chain consists of $N$ identical (unit) masses connected by nonlinear springs to their nearest neighbours. The dynamics is generated by the Hamiltonian
\begin{equation}
\label{eq:FPUHam}
	H=\;\sum_{j\in \mathbb{Z}_N} \left[\frac{p_j^2}{2}+\phi(q_{j+1}-q_j) \right]\ ,
\end{equation}
where $\mathbb{Z}_N:=\mathbb{Z}/(N\mathbb{Z})$, and $\phi$ is the potential
\begin{equation}\label{eq:PhiPotential}
	\phi(z) \;:=\; \frac{z^2}{2}+\alpha \frac{z^3}{3}+\beta \frac{z^4}{4}+O(z^5)\ ,
\end{equation}
and $\alpha$,$\beta$,\dots are the parameters measuring the strength of the nonlinear terms. One usually refers to the \emph{$\alpha$-model} if $\alpha$ is the only non-zero parameter; to the \emph{$\beta$-model} if $\beta$ is the only non-zero parameter; to the \emph{$\alpha+\beta$-model} if both $\alpha$ and $\beta$ are non-zero, and to the \emph{generalised FPU model} if the lowest degree of the nonlinearity is greater than  or equal $5$.

When all the parameters in the nonlinearity are set to zero, the Hamiltonian (\ref{eq:FPUHam}) reduces to that of a \emph{harmonic chain}, where particles interact through linear forces only. The latter system is integrable in the sense of Liouville, and the Hamiltonian is diagonalised by the (discrete) Fourier transform
\begin{equation}
	p_j \;=\; \frac{1}{\sqrt{2N}} \sum_{k=-N}^N \hat{p}_k e^{\imath \pi \frac{j k}{N}}\ ,
\end{equation}
and similarly for $q_j$. The integrals of motion are the energies of the Fourier modes
\begin{equation}
	E_k \;=\; \frac{|\hat{p}_k|^2+ \omega_k^2 |\hat{q}_k|^2}{2} \, , \qquad k=-N,\dots,N-1\ ,
\end{equation}
where $\omega_k:=2 \big| \sin\big(\frac{k \pi}{2N} \big) \big|$ are the proper frequencies of oscillation. Observe that $E_k=E_{-k}$, for all $k$.

The nonlinear model (\ref{eq:FPUHam}) was introduced by Fermi, Pasta and Ulam (FPU), supported by Tsingou \cite{FPU1955}, with the purpose of analysing its thermalisation process. The authors expected that the interaction between the Fourier modes due to the nonlinear terms, and the consequent energy sharing between them, would have brought the system to reach the thermal equilibrium on a short time-scale. In particular, as a detector of thermal equilibrium, they expected to observe the ``equipartition of energy'', i.e. a final state of the system where, on time-average, all Fourier energies have almost the same value, i.e. $E_k \simeq E/N$, where $E$ is the total energy. Their numerical simulations showed instead a completely different scenario: by initially exciting the lowest frequency mode ($k=1$), within their available computation time, energy sharing was observed to effectively take place only among the first few modes and, instead of a continuous trend to equipartition, the dynamics showed an almost recurrent behaviour. The first explanation of the latter phenomenon goes back to Zabusky and Kruskal \cite{ZK}, who approximated the traveling wave dynamics of the system by the KdV equation, and based their argument on the recurrent behaviour of its solitons. On the Hamiltonian side, the first correct computation of 
the resonant normal form of the lattice system, in action angle-variables, is due to \cite{Shepelyansky}. Such a construction was only later recognised to include that of Zabusky and Kruskal \cite{BamPo,BLP}. 

Nowadays, it is well known that a key role in the explanation of the FPU phenomenon, or paradox, is played by the \emph{integrability} of the resonant normal form either of the lattice system or of its infinite-dimensional approximation (we refer to \cite{Gallavotti-2008,Bambusi2015} and references therein). Indeed, the KdV equation admits a complete set of (infinitely many) integrals of motion, whose conservation prevents a fast energy sharing among the Fourier modes. Moreover, the preservation of the analyticity of the initial condition causes an exponential decay of its Fourier energies \cite{KaPo}. These two aspects resemble very much the observations in the FPU experiment. 

In fact, the connection FPU-KdV can be made rigorous using the normal form construction of Theorem \ref{prop:PerturbationGeneral}, as follows. As a preliminary step, we perform the canonical change of variables $(q,p) \mapsto (s,r)$ defined by the generating function
\begin{equation}
	F(q,s) \;=\;\sum_{j \in \mathbb{Z}_N} s_j(q_j-q_{j+1})\ ,
\end{equation}
which gives
\begin{equation}
	\begin{split}
		r_j \;&= \;-\frac{\partial F}{\partial s_j}=q_{j+1}-q_j \, ,\\
		p_j \;&=\; \frac{\partial F}{\partial q_j}=s_j-s_{j-1} \, .
	\end{split}
\end{equation}
In terms of the new variables $(s,r)$ the Hamiltonian (\ref{eq:FPUHam}) reads
\begin{equation}
H=\sum_{j\in\mathbb{Z}_N}\left[\frac{(s_{j+1}-s_j)^2}{2}+\phi(r_j)\right]\ ,
\end{equation}
whose equations of motion are
\begin{equation}
\label{eq:srjdot}
\begin{split}
		\dot{s}_j\ &=\frac{\partial H}{\partial r_j}=\phi'(r_j)\ ,\\
	         \dot{r}_j\ &=-\frac{\partial H}{\partial s_j}=s_{j+1}+s_{j-1}-2s_j\ .
	\end{split}
\end{equation}
\begin{remark}
The periodicity of the $q_j$ implies $\sum_{j\in\mathbb{Z}_N}r_j=0$, whereas the periodicity of the $s_j$ implies
$\sum_{j\in\mathbb{Z}_N}p_j=0$.
\end{remark}
We now assume the existence of a pair of analytic functions $R,S:\mathbb{T} \times \mathbb{R} \to \mathbb{R}$ such that
\begin{equation}
\label{eq:srSR}
	\begin{split}
			s_j(t)\;&=\; \dfrac{\sqrt{\varepsilon}}{h}S(x,\tau)|_{x=hj,\tau=ht} \, \\
			r_j(t) \;&=\; \sqrt{\varepsilon} R(x,\tau) |_{x=hj,\tau=ht} \, .
	\end{split}
	\ \ \ \ \ h:=\frac{1}{N}\ .
\end{equation}
Notice that the choice of the functions $R$ and $S$ is \emph{not} unique. For example, one can add to them any
linear combination of the form $\sum_{m\in\mathbb{Z}}c_m\sin(\pi mx/h)$, which vanishes at the lattice sites
$x=hj$.  Having in mind long-wavelength initial conditions, a natural choice consists in restricting $R$ and $S$ to the Fourier polynomials supported on the first few harmonics
at $\tau=0$, and in regarding the discrete system as a sampling of the continuous one at any $\tau>0$. This is allowed by the following proposition.
\begin{proposition}
Consider the Hamiltonian functional
\begin{equation}
\label{eq:ContinuousFPU}
	\mathcal{H}[S,R]=\oint \Big[\frac{1}{\varepsilon} \phi(\sqrt{\varepsilon} R(x,\tau))-\frac{1}{2} S(x,\tau) \Delta_h S(x,\tau) \Big] \, dx\ ,
\end{equation}
where 
\begin{equation}
\label{eq:Deltah}
	\Delta_h:= \frac{4}{h^2}\sinh^2\left(\frac{h}{2}\partial_x\right)= \partial_x^2+\frac{h^2}{12} \partial_x^4 + O(h^4)
\end{equation}
is the discrete Laplacian. Then, its Hamilton equations restricted to the the lattice coincide with the FPU equations
(\ref{eq:srjdot}). 
\end{proposition}
\begin{proof}
Considering $(S,R)$ as a canonical pair coordinate-momentum, one has
\begin{equation}
\label{eq:SRtau}
\begin{split}
S_\tau=&\frac{\delta\mathcal{H}}{\delta R}=\frac{1}{\sqrt{\varepsilon}}\phi'(\sqrt{\varepsilon}R)\ ,\\
R_\tau=&-\frac{\delta\mathcal{H}}{\delta S}=\Delta_h S\ .
\end{split}
\end{equation}
The latter equations, restricted to the lattice, i.e. to the points $x=hj$, coincide with those obtained by substituting
(\ref{eq:srSR}) into (\ref{eq:srjdot}). \qed
\end{proof}

One thus embeds the dynamics of the FPU lattice (\ref{eq:srjdot}) within that of the infinite-dimensional Hamiltonian system
(\ref{eq:SRtau}). The latter consists of a system of nonlinear dispersive Hamiltonian PDEs for any expansion to finite order of the discrete Laplacian (\ref{eq:Deltah}). Moreover, making use of the latter expansion and of the explicit expression 
\eqref{eq:PhiPotential} of $\phi$, one observes that the Hamiltonian \eqref{eq:ContinuousFPU} has the grading of Definition \ref{def:DefGrading} with 
\begin{equation}
\label{eq:fpugrad}
\lambda \sim \sqrt{\varepsilon} \sim h^2\ .
\end{equation}

Let us see in which sense KdV equation allows us to explain rigorously, in the case of the $\alpha$-chain, the FPU phenomenon, namely the fact that, if one low-frequency mode is initially excited, then the energy quickly flows to a small packet of modes whose energy, on time average, decreases exponentially with the mode index. The main result is conveniently formulated in terms of the 
quantities 
\begin{equation}
	\kappa\;:=\;\frac{k}{N}\ \ ;\ \ \mathcal{E}_\kappa \;:=\; \frac{E_k}{N}\ ,
\end{equation}
denoting the specific mode index (or wave number) and the corresponding specific energy, respectively. 
We are interested in the evolution of initial data supported on one harmonic mode of long wavelength, i.e. specific index
$\kappa_0=k_0/N\ll 1$.

\begin{theorem}[Bambusi-Ponno \cite{BamPo}]\label{thm:LocalisationFPU}
Consider an initial condition of the form
\begin{equation}
\mathcal{E}_{\kappa_0} (0)\; = \; C_0\mu^4 , \qquad \mathcal{E}_\kappa(0) =\;0 , \qquad \forall \kappa 	\neq \kappa_0 \, ,
\end{equation}
where $C_0$ is any fixed constant and $\mu:=\kappa_0:=k_0/N\ll 1$. 

Then, for any fixed time $T_f$ there exist positive
constants $\mu^*$, $\sigma$, $C_1$ and  $C_2$ (dependent on $C_0$ and $T_f$) such that, 
for all $\kappa$, $\mu<\mu^*$ and $|t|\leq T_f/\mu^3$
\begin{enumerate}[label={(\roman*)}]
	\item 
	\begin{equation}
	\mathcal{E}_\kappa(t)\; \leq\; \mu^4 C_1 e^{- \sigma \kappa/\mu}\ ;
	\end{equation}
	\item there exists a sequence of almost periodic functions $\{F_n(t)\}_{
n\in \mathbb{N}}$ and an associated specific sequence 
\begin{equation}
F_\kappa=\mu^4F_n\ ,\quad  \text{if}\quad  \kappa=n\kappa_0\ ;\quad F_\kappa =0 \quad \text{otherwise}\ ,
\end{equation}
such that
\begin{equation}
|\mathcal{E}_\kappa(t) - F_\kappa(t)| \;\leq\; C_2 \mu^5\ .
\end{equation}
\end{enumerate}

\end{theorem}

The proof of this theorem is based on the fact that a solution of the KdV equation with an analytic initial datum on the torus remains analytic for all times \cite{KaPo}. In particular, analyticity implies the exponential decay of Fourier coefficients, which in turn implies the exponential decay of the Fourier coefficients for the FPU system. 

On the other hand, technical difficulties arise when comparing the dynamics of the discrete system with the dynamics of the continuous one,  due to the contribution of the \emph{singular} remainder of the discrete Laplacian that contains higher-order derivatives. The latter problem is overcome by a combined use of the analyticity property of the KdV flow, closeness to the identity of the canonical transformation and Gr\"onwall lemma \cite{BamPo}.

However, when comparing the above result with the numerical simulations and with the recent results on relating the FPU dynamics to that of the Toda lattice \cite{BamMasToda}, one realises that it is not optimal: the time scale of closeness to the KdV dynamics numerically observed turns out to be longer than $t\sim\mu^{-3}\sim\varepsilon^{-3/4}$.  In fact, there is an actual hope to improve the latter result which rests on the fact that the normal form of the FPU problem is in the \emph{KdV hierarchy} not only to the first, but also to the second perturbative order. 
%In order to show this, we first shortly recall  that the KdV equation is an infinite dimensional, integrable, Hamiltonian system displaying infinitely many constants of motion in involution (i.e. functionals on the phase space with pairwise vanishing Poisson bracket). The explicit form of the first three integrals is
%\begin{eqnarray}
%	I_0&=&\oint u^2 \, dx \, ,\\
%	I_1&=&\oint \Big(\frac{1}{3}\gamma u^3+uu_{xx} \Big) \, dx \, ,\\
%	I_3&=&\oint \Big(\frac{5}{36} \gamma^2 u^4 + \frac{5}{6} \gamma u^2 u_{xx} + (u_{xx})^2 \Big) \, dx \, ,
%\end{eqnarray}
%where $\gamma$ is a real parameter.  
Then, an extension of Theorem \ref{thm:LocalisationFPU} could work with a second-order normal form transformation yielding the (presumably) optimal result of localisation of the Fourier spectrum on time-scales 
$\sim \mu^{5}\sim \varepsilon^{-5/4}$.

Within this context, we present below the normal form construction of the FPU problem, including the Kodama transformation.
 
\begin{proposition}\label{prop:FPUKDV}
The Hamiltonian \eqref{eq:ContinuousFPU} can be mapped into the normal form
	\begin{equation}
	\label{eq:hnffpu}
		\tilde{H}=K_0+Z_1+Z_2+\dots\ ,
	\end{equation}
	with
	\begin{eqnarray}
		K_0&=&\oint \frac{{\tilde{u}}^2+{\tilde{v}}^2}{2} \, dx		\label{eq:H0FPU}\\
		Z_1&=&\frac{h^2}{4! 2} \oint \left[\frac{4 \alpha \sqrt{2 \varepsilon}}{h^2}\big({\tilde{u}}^3+{\tilde{v}}^3 \big)+{\tilde{u}}{\tilde{u}}_{xx}+{\tilde{v}}{\tilde{v}}_{xx} \right] \, dx  \label{eq:Z1FPU}\\
		Z_2&=&{\textstyle\frac{3}{20} \frac{h^4}{(4!)^2}}\oint \Big[{\textstyle\left(\frac{\beta}{\alpha^2}-\frac{1}{2}\right) \frac{240 \alpha^2 \varepsilon}{h^4}}({\tilde{u}}^4+{\tilde{v}}^4)+{\textstyle\frac{20 \alpha \sqrt{2 \varepsilon}}{h^2}}\big({\tilde{u}}^2{\tilde{u}}_{xx}+{\tilde{v}}^2{\tilde{v}}_{xx} \big) \nonumber\\ &
		\,&+({\tilde{u}}_{xx})^2+({\tilde{v}}_{xx})^2 \Big] \, dx + \Big({\textstyle\frac{3 \beta \varepsilon}{8}-\frac{\alpha^2 \varepsilon}{4} }\Big) \Big( \oint {\tilde{u}}^2 \, d x \Big) \Big(\oint {\tilde{v}}^2 \, d x \Big) +\nonumber \\ &\,&+\frac{\alpha^2 \varepsilon}{32} \Big(\langle {\tilde{u}}^2 \rangle^2+\langle {\tilde{v}}^2 \rangle^2 \Big) \label{eq:Z2FPU}
	\end{eqnarray}
\end{proposition}

\begin{proof}
By introducing the Riemann variables 
\begin{equation}
\label{eq:uvFPU}
u:=\frac{S_x+R}{\sqrt{2}}\ \ ;\ \ v:=\frac{S_x-R}{\sqrt{2}}\ ,
\end{equation}
the result is actually a Corollary of Proposition \ref{prop:GenericNormalForm}, with the substitutions $\alpha_1=\frac{\alpha \sqrt{\varepsilon}}{6 \sqrt{2}}$, $\alpha_2=-\frac{h^2}{4! 2}$ and $\beta_1=\frac{\beta \epsilon}{4}$. \qed
\end{proof}
\begin{remark}
The equations of motion of $K_0+Z_1$ are those of two counter-propagating KdV equations, i.e. of the form (\ref{eq:ord2}), for any $\alpha$. On the other hand, the equations of motion of $K_0+Z_1+Z_2$ are not in the KdV hierarchy, i.e. in the form
(\ref{eq:ord4}), unless the special condition $\beta=5\alpha^2/6$ holds.
\end{remark}
%However, one can perform a further close-to-identity transformation of the vector field $J\nabla(K_0+Z_1+Z_2)$ into a vector field of the KdV hierarchy, i.e. of the form of the RHS of (\ref{eq:ord4}). To the authors' knowledge, such a procedure dates back to Kodama \cite{Kod85,Kod87a,Kod87b,Kodama2} for the transformation of vector fields defined on the line, and it has been extended to the torus only recently \cite{GPR}.
In order to bring the continuous FPU equations of motion into the KdV hierarchy form (\ref{eq:ord4}), one must look for a suitable Kodama transformation, as sketched in Remark \ref{rem:Kodama}.
\cite{GalloneLM}.
\begin{proposition}
The Kodama transformation
	\begin{equation}
			\tilde u = w + g(w)\ \ ;\ \ \tilde v = z + g(z)\ ,
	\end{equation}
	where
	\begin{equation}
	\label{eq:gkodama}
	\begin{split}
		g(w):=&\frac{h^2}{4!}\left(\frac{7}{2}-\frac{9}{2}\frac{\beta}{\alpha^2}\right)w_{xx}+
		\frac{\alpha\sqrt{\varepsilon}}{\sqrt{2}}\left(\frac{13}{12}-\frac{3}{2}\frac{\beta}{\alpha^2}\right)
		\left(w^2-\oint w^2\ dx\right) + \\
		 -&\frac{1}{6}\left(w_x\partial_x^{-1}w-\oint w^2\ dx\right)\ ,
		\end{split}
	\end{equation}
maps the equations of motion of the Hamiltonian normal form (\ref{eq:hnffpu}) into the integrable KdV form (\ref{eq:ord4}).
\end{proposition}
\begin{proof}
The proof consists in a long, though direct computation. Details can be found in \cite{GalloneLM}. Observe that, according to the grading
(\ref{eq:fpugrad}), $g\sim \lambda$, which does not affect the first order normal form. \qed
\end{proof}

A natural question arises now, namely whether it is possible to construct a normal form transformation, including the Kodama procedure, conjugating the continuous FPU equations to those of the KdV hierarchy to perturbative orders higher than the second one.  This is an open problem for initial data generically supported on lower modes, but it has recently been addressed for initial data close to the traveling wave. In \cite{GPR} it is proved that for almost-traveling waves, the conjugation to the third-order works only if the parameters correspond to a curve in the space of parameters containing the Toda lattice.

In general, it is expected that the FPU normal form is in the KdV hierarchy to a finite perturbative order, depending on the model. This is easily seen by considering the family of generalised FPU-systems \cite{PonBam08} defined by a Hamiltonian of the form \eqref{eq:FPUHam} with
\begin{equation}
\label{eq:genFPU}
	\phi(z)=\frac{z^2}{2}+\frac{z^p}{p} \, , \qquad p \geq 3\ .
\end{equation}
Instead of fixing a model and going on with the perturbative order, we here consider how the first order normal form depends on the exponent $p$.
%Since, as shown above, integrability of the resonant normal form plays an important role in the explanation of the FPU phenomenon, one is naturally led to the study of the $\beta$-model and systems with higher degrees of nonlinearity to check the occurrence and relevance of integrability.
The Hamiltonian (\ref{eq:ContinuousFPU}) with potential (\ref{eq:genFPU}) reads
\begin{equation}
	H \;=\; \oint \Big[\frac{R^2}{2}+\gamma \varepsilon^{\frac{p-2}{2}} \frac{R^p}{p}+\frac{1}{2} (S_x)^2-\frac{h^2}{12}(S_{xx})^2 \Big] \, dx + O(h^4)\ .
\end{equation}
Passing to the $(u,v)$ variables (\ref{eq:uvFPU}), one gets $H=K_0+H_1$, where
\begin{eqnarray}
	K_0&=&\oint \frac{u^2+v^2}{2} \, dx\ ,\\
	H_1&=&\oint \left[\gamma \varepsilon^{\frac{p-2}{2}} \frac{(u-v)^p}{2^{p/2} p}-\frac{h^2}{24}\big((u_{x})^2 + 2 u_x v_x + (v_x)^2\big) \right] \, dx\ .
\end{eqnarray}
Averaging $H_1$ (using \eqref{eq:IntegralUV}) one computes the normal form $\tilde H=K_0+Z_1+\cdots$ of the system, where
\begin{equation}
	\begin{split}
	Z_1= \left\langle H_1\right\rangle\;&=\; \oint \frac{\gamma \epsilon^{\frac{p-2}{2}}}{2^{p/2} p} \big(u^p+(-v)^p\big)-\frac{h^2}{24} \big((u_x)^2+(v_x)^2 \big) \, dx \\
	& + \frac{\gamma \epsilon^{\frac{p-2}{2}}}{2^{p/2} p} \sum_{j=1}^{p-1} (-1)^j \binom{p}{j} \Big(\oint u^{p-j} \, dx \Big) \Big(\oint v^j \, dx \Big) \, .
	\end{split}
\end{equation}
For $p=3$ one finds that $K_0+Z_1$ is the Hamiltonian of two uncoupled KdV equations, as expected. For $p=4$, the so-called 
$\beta$-model, $K_0+Z_1$ is the Hamiltonian of two uncoupled modified KdV (mKdV) equations. Thus, the first order normal form is integrable for $p=3,4$. On the other hand, for $p\geq 5$ the first order normal form Hamiltonian is that of two generalised, nonintegrable KdV equations, that are also nonlinearly coupled for $p\geq6$. For this class of models the integrability of the normal form, and the consequent FPU phenomenon of energy localisation due to closeness to integrability, are lost to first order if the degree of nonlinearity is high enough ($p\geq5$). More than this, in \cite{PonBam08} it is suggested that the blow-up of solutions characterising the nonintegrable KdV equations might play a relevant role in the problem.
%The situation changes with higher-order FPUs where the normal form equations one gets are higher-order generalized KdVs (gKdV) that are no longer integrable. One is thus naturally led to ask whether the FPU phenomenon persists in such cases, and whether the integrability of the normal form plays a fundamental role or not. In \cite{PonBam08} it is suggested that the answer could be no and that there should be a weaker property entraining the formation of a metastable packet of modes. It is guessed that such a property is related to the smoothness of the solutions to the Cauchy problem of the normal form PDEs. In particular, if the solutions of the PDEs under investigation blow up in a finite time then metastability should be lost, otherwise, the phenomenon of formation of a metastable packet of modes should be present in all models that do not display blow-up. For the case of one-dimensional generalised FPU chains with potential \eqref{eq:genFPU} it turns out that the metastable packet is expected to exist in the case of nonlinearity of degree $p \leq 5$, while it is expected not to exist when the nonlinearity has degree $p\geq 7$. The case of nonlinearity of order $6$ is critical and therefore the existence of the packet strongly depends, in principle, on the initial conditions.

As a last point, we stress that the method of infinite-dimensional perturbation theory allows to analyse the FPU system, treated in Proposition \ref{prop:FPUKDV}, in the singular limit $h \to 0$ with fixed, small specific energy $\varepsilon$. 
Such a limit is justified on the short term, where dispersion is expected to play a minor role with respect to nonlinearity, which explains why the normal modes start to effectively share their energy. Taking the limit $h\to0$, at fixed $\varepsilon$, of the FPU terms (\ref{eq:H0FPU}), (\ref{eq:Z1FPU}) and (\ref{eq:Z2FPU}), one finds
\begin{equation}\label{eq:HamFPUNoDisp}
	H=K_0+Z_1+Z_2+\dots
\end{equation}
with
\begin{eqnarray}
	K_0&=&\oint \frac{u^2+v^2}{2} \, dx \, , \\
	Z_1&=& \frac{\alpha \sqrt{\varepsilon}}{2\sqrt{2}} \oint \frac{u^3+v^3 }{3} \, dx \, ,\\
	Z_2&=&\Big(\frac{\beta}{\alpha^2}-\frac{1}{2} \Big) \frac{\alpha^2 \varepsilon}{4} \oint \frac{u^4+v^4}{4} \, dx \, . \label{eq:LastNoDisp}
\end{eqnarray}
The equations of motion associated to this normal form Hamiltonian consist of a pair of uncoupled, generalised Burgers equations, whose solution displays a gradient catastrophe at a finite shock time $t_s$. It has recently been proved that
the Fourier energy spectrum of such a system displays a power law decay characterised by the universal exponent 
$-8/3$ exactly at $t_s$. Such a prediction fits very well the numerical spectrum of the FPU system \cite{GallonePonnoRuffo}. Of course, the dynamics on times longer than $t_s$ cannot be described in this limit and dispersive effects
must be re-included, in agreement with the grading (\ref{eq:fpugrad}).  
%These equations are useful in the study of the Fermi-Pasta-Ulam system for $h\ll \varepsilon <O(1)$ and for short times. A crucial problem is the fact that normal form equation admits solution only for a finite time interval $t \in [0,t_s)$, while original equations of motion admits solutions for all times. Clearly, this is  a limitation in approximating the 'true' dynamics, but one may be able to obtain interesting information on the system in the regime of specific energy where the metastable state could eventually be spread or for the dynamics of long-wavelength Fourier modes for which dispersion is not relevant.

\subsection{Water waves}\label{subsec:WW}
Consider an ideal fluid occupying, at rest, the domain
\begin{equation}
	\Omega_{0,L}\;:=\; \big\{(x,z) \in [0,L] \times \mathbb{R} \,:\, -h<z<0 \big\} \, ,
\end{equation}
with $L>0$. We study the evolution of the free surface under the action of gravity, in the irrotational regime. Thus, given a periodic function $\eta:[0,L] \to \mathbb{R}$, we define the domain
\begin{equation}
	\Omega_{\eta,L} \;:=\;\big\{(x,z) \in [0,L] \times \mathbb{R} \, : \, -h<z<\eta(x) \big\} \, .
\end{equation}
Irrotationality makes it possible to describe the velocity of the fluid $u$ as gradient of a function called \emph{velocity potential} by $u=\nabla \phi$. This problem admits a Hamiltonian formulation \cite{Zak68,CS93,CS94} and the conjugated variables are the wave profile $\eta(x)$ and the trace of the velocity potential at the free surface:
\begin{equation}
	\psi(x) \;:=\; \phi(x,\eta(x)) \, .
\end{equation}
The Hamiltonian of the system is
\begin{equation}\label{eq:WaterWavesHam}
	H(\eta,\psi) \;=\; \oint \Big(\frac{1}{2}g \eta^2+\frac{1}{2} \psi G(\eta) \psi \Big) \, dx
\end{equation}
where $G(\eta)$ is the \emph{Dirichlet-to-Neumann operator} defined as follows. Given a function $\psi(x)$ and consider the boundary value problem 
\begin{eqnarray}
\Delta \phi &=&0 \, , \qquad (x,z) \in \Omega_{\eta,L} \\
\phi_z\Big|_{z=-h}&=&0 \\
\phi(0) &=& \phi(L) \label{eq:LateralBoundaryWW}\\
\phi \Big|_{z=\eta(x)} &=& \psi
\end{eqnarray}
and let $\phi$ be its solution. Then
\begin{equation}
	G(\eta)\psi\;:=\; \sqrt{1+\eta_x^2} \partial_n \phi\big|_{z=\eta(x)} =(\phi_z-\eta_x\phi_x)\big|_{z=\eta(x)}
\end{equation}
where $\partial_n$ denotes the derivative in the direction normal to $z=\eta(x)$.

We are interested in solutions of the form
\begin{equation}
	\eta(x)\;=\;\mu^2 h^3 \sqrt{2} \tilde{\eta}(\mu x) \, , \qquad \psi(x)\;=\; \mu \sqrt{2 g h} h^2 \tilde{\psi}(\mu x) \, , \qquad \mu=1/L \ll 1 \, ,
\end{equation}
that corresponds to a canonical transformation when rescaling time to
\begin{equation}
	\tilde{t} \;=\; \frac{t}{\mu \sqrt{gh}} \, 
\end{equation}
and the physical space becomes the torus of unitary length.

Note that the dependence on $\eta$ of the Dirichlet-to-Neumann operator causes the Hamiltonian \eqref{eq:WaterWavesHam} not to fall within the class of mechanical Hamiltonians of Subsection \ref{subsec:MechanicalHamiltonians}.

The small parameter of the theory is $\lambda = (h \mu)^2$. Expanding the Hamiltonian in $\lambda$ one gets\footnote{This step is far from being a trivial Taylor expansion as it involves the asymptotic expansion of the Dirichlet-to-Neumann operator (see \cite{BamWW} for details).}
\begin{equation}\label{eq:HamiltonianWWTBN}
	H\;=\;H_0+\lambda H_1 + \lambda^2 H_2 + O(\lambda^3)
\end{equation}
with $H_0$ being in the same form of \eqref{eq:H0FPU} but with renamed variables:
\begin{eqnarray}
	H_0 &=& \oint \frac{\tilde{\eta}^2+\tilde{\psi}_y^2}{2} \, dy \, ,\\
	H_1&=& \frac{1}{2} \oint \Big(-\frac{1}{3}\tilde{\psi}_{yy}^2+\sqrt{2} \tilde{\eta} \tilde{\psi}_y^2 \Big) \, dy \, ,\\
	H_2 &=& \frac{1}{2} \oint \Big(\frac{2}{15} \tilde{\psi}_{yyy}^2-\sqrt{2} \tilde{\eta} \tilde{\psi}_{yy}^2 \Big) \, dy \, .
\end{eqnarray}
Note that, the Hamiltonian contains terms with the product of $\tilde{\eta}$ and $\tilde{\psi}$ and thus doesn't fit the definition of \emph{mechanical Hamiltonian} given above. Anyway, as for the FPU problem, it is convenient to use characteristic variables $(u,v)$ defined as
\begin{eqnarray}
	\tilde{\eta}(y,t) &=& \frac{u(y,t)+v(y,t)}{\sqrt{2}} \, , \\
	\tilde{\psi}_y(y,t) &=& \frac{u(y,t)-v(y,t)}{\sqrt{2}}
\end{eqnarray}
we then obtain
\begin{eqnarray}
	K_0&=&\oint \frac{u^2+v^2}{2} \, dy \, ,\\
	H_1&=& \oint \Big(-\frac{1}{12}(u_y^2+v_y^2)+\frac{u^3+v^3}{4}+\frac{u_yv_y}{6}-\frac{u^2v+uv^2}{4} \Big) \, dy \, , \\
	H_2&=&\oint\Big( \frac{1}{2} \frac{u_{yy}^2+v_{yy}^2}{15}-\frac{1}{4} (u u_y^2+ v v_y^2)-\frac{1}{15} u_{yy} v_{yy} \\
	& \, &-\frac{1}{4} (uv^2_y-2 u u_y v_y+vu_y^2-2vu_yv_y) \Big) \, dy \, .
\end{eqnarray}
Applying the techniques of Theorem \ref{prop:PerturbationGeneral} one has

\begin{proposition}
	Within the normal form procedure outlined above, Hamiltonian \eqref{eq:HamiltonianWWTBN} can be mapped into the normal form
	\begin{equation}
		\tilde{H} \;=\; \tilde{K}_0+\lambda Z_1+ \lambda^2 Z_2+\dots
	\end{equation}
	with
	\begin{eqnarray}
		Z_1&=& \oint \Bigg[\frac{\tilde{u}^3+\tilde{v}^3}{4}-\frac{1}{12} \big(\tilde{u}_y^2+\tilde{v}_y^2\big)\Bigg] \, d y \label{eq:S1-water-waves} \, ,\\
		Z_2&=& \oint \Bigg[\frac{1}{64} \big(\tilde{u}^4+\tilde{v}^4 \big)+\frac{7}{48} \big(\tilde{u}^2 \tilde{u}_{yy}+\tilde{v}^2 \tilde{v}_{yy} \big) + \frac{29}{720} \big( \tilde{u}_{yy}^2+\tilde{v}_{yy}^2\big) \Bigg] \, d y \nonumber \\
		& &\quad + \frac{1}{8} \langle \tilde{u}^2 \rangle \langle \tilde{v}^2  \rangle \, .\label{eq:S2-water-waves}
	\end{eqnarray}
\end{proposition}	

\begin{proof} This result is proved computing normal form Proposition \ref{prop:GenericNormalForm}. 

\noindent
	\emph{First perturbative step:} We use \eqref{eq:IntegralUV} to average $H_1$ along the flow of $H_0$ obtaining the expression for $Z_1$ in \eqref{eq:S1-water-waves}. The Hamiltonian generating the canonical transformation can can be computed using \eqref{eq:IntegralSUV}:
	\[
		G_1\;=\; -\oint \bigg[\frac{1}{12}v_yu+\frac{1}{8}u^2 \partial_y^{-1} v - \frac{1}{8} v^2 \partial_y^{-1} u \bigg] \, dy \, .
	\]
	We can therefore compute the $L_2$-gradient of $G_1$ and of $H_1-Z_1$ obtaining
	\[
		\begin{split}
			\frac{\delta G_1}{\delta u} \;&=\; -\frac{1}{12} v_y - \frac{1}{4} u \partial_y^{-1} v - \frac{1}{8} \partial_y^{-1} v^2 \, ,\\
			\frac{\delta G_1}{\delta v} \;&=\; \frac{1}{12} u_y + \frac{1}{8} \partial_y^{-1} u^2 + \frac{1}{4} v \partial_y^{-1} u \, ,
		\end{split}
	\]
	\[
		\begin{split}
			\frac{\delta (H_1-Z_1)}{\delta u} \;&=\; - \frac{1}{6} v_{yy}-\frac{1}{2} uv -\frac{1}{4} v^2 \, , \\
			\frac{\delta (H_1-Z_1)}{\delta v} \;&=\; - \frac{1}{6} u_{yy}-\frac{1}{2} uv - \frac{1}{4} u^2 \, .
		\end{split}
	\]

\noindent	
	\emph{Second perturbative step:} We use \eqref{eq:IntegralUV} to average $H_2$ and $\{H_1-Z_1,G_1\}$ obtaining:
	\[
		\begin{split}
			\big\langle \{H_1,G_1\} \big\rangle_{0} \;&=\; \frac{1}{8} \oint \Big[\frac{1}{9}(u_{yy}^2+v_{yy}^2)+\frac{1}{3}(u^2 u_{yy} +v^2 v_{yy})+\frac{1}{4}(u^4+v^4) \Big] \, dy \\
			&\qquad \frac{1}{4} \langle u^2 \rangle \langle v^2 \rangle \, ,
		\end{split}
	\]
	\[
		\langle H_2 \rangle_0 \;=\; \oint \Big[\frac{u_{yy}^2+v_{yy}^2}{30}+\frac{1}{8} \big( u^2 u_{yy}+v^2 v_{yy} \big) \Big] \, dy	\, .
	\]
	We obtain $Z_2=\langle H_2 \rangle_0 + \frac{1}{2} \langle \{H_1,G_1 \} \rangle_0$ that is precisely \eqref{eq:S2-water-waves}.
\qed
\end{proof}
As for equations of the FPU lattice, these Hamiltonians are \emph{not} in the Korteweg-de Vries hierarchy. Exactly as in the previous case, Kodama's theory solves the problem and with a close-to-identity change of variables maps the Hamiltonian into:
\begin{equation}\label{eq:WWNF}
	H_{\mathrm{NF}}(u,v)\;=\;K_0(u)+\lambda K_1(u)+\lambda^2 c_2 K_2(u) +K_0(v)+\lambda K_1(v)+\lambda^2 c_2 K_2(v) \, 
\end{equation}
with $c_2$ being some explicit constant.

In case $\mu$ is a small free parameter not related to $L$ and the water waves are studied on the whole real line (that is, $x \in \mathbb{R}$ and thus, imposing $\lim_{x \to \infty} \phi(x)=0$ instead of $\phi(0)=\phi(L)$ in \eqref{eq:LateralBoundaryWW}), the following result holds
\begin{theorem}[Bambusi \cite{BamWW}] \label{thm:BamWW}
	For any $s'$, there exists $\lambda^*>0$ and $s,s''$, s.t., if $0<\lambda < \lambda^*$, then there exists a map $T_\lambda:B_1^s \to W^{s'',1} \times W^{s'',1}$, with the following properties
\begin{itemize}
	\item[(i)] $\sup_{(u,v) \in B_1^s} \Vert T_\lambda(u,v)-(u,v) \Vert_{W^{s'',1} \times W^{s'',1}} \leq C \lambda$,
	\item[(ii)] Let $I_\lambda$ be an interval containing the origin and $z(\cdot)=(u(\cdot),v(\cdot)) \in C^1(I_\lambda;B_1^s)$ be a solution of the Hamiltonian system \eqref{eq:WWNF} with $c_2=\frac{299}{389}$ define
	\begin{equation}
		z_a=(u_a,v_a)\;:=\;T_\lambda(u,v) \, .
	\end{equation}
	Then there exists $R \in C^1(I_\lambda,W^{s',2}\times W^{s',2})$ s.t. one has
	\begin{equation}
		\dot{z}_a \;=\; J \nabla H (z_a(t))+\lambda^3 R(t) \, \qquad \forall t \in I_\lambda \, ,
	\end{equation}
	where $H$ is the Hamiltonian of water waves problem in the variables $u$ and $v$.
\end{itemize}
\end{theorem}

An interesting non-trivial dynamical information one can obtain from this Theorem concerns the goodness of the approximation of the normal-form dynamics. That is, for smooth enough initial data, it is possible to go back to the original non-scaled variables and to get the estimate on the wave profile
\begin{equation}
	\sup_{|t| \leq T^*/{\mu^3 \sqrt{gh}}} \Vert \eta(t)-\eta_a(t) \Vert_{L^\infty} \leq C\mu^6 \, .
\end{equation}

Note that the difference between wave profiles can be proved to be small only for times in which the second perturbative correction is negligible. Thus, as for the FPU system, an interesting open problem is the understanding which results can hold for larger time scales.

We are confident that these two results can be proved also in the periodic setting presented above.

\section{Conclusions and open problems}

In the framework of Hamiltonian field theory, the continuum limit of the FPU chain for long-wavelength excitations and the Hamiltonian of water waves belong to the same wider class of perturbations of the wave equation. This is not the case of other lattice models, such as the Klein-Gordon, for which one has to take into account the presence of the mass term.

Recently, the analysis of lattice model using the machinery of water waves has received a certain interest especially for systems in two spatial dimensions \cite{PelinovskiHir,GaPa} or for the analysis of higher-order normal forms for one-dimensional systems \cite{GPR}. 

As a comparison, water waves are now a hot topic in research. The main goals in the field are results on well-posedness as well as regularity result for solutions or existence of quasi periodic or traveling wave solutions (see e.g. the recent results \cite{BBHM,FeolaGiuliani,BertiFranzoiMaspero,Berti2021}).

In this sense, many open questions remain open and can hopefully be addressed in the next future:
\begin{itemize}
	\item The analysis at second order performed in Subsec. \ref{subsec:FPU}  does not allow us to conclude that the dynamics of the integrable system is close
to the dynamics of the original system. Actually, it is known how to obtain a result on the dynamics, but only over times over which the effects of the second-order term is invisible. One of the open major problems is to understand how to go beyond the time scale of Theorem \ref{thm:LocalisationFPU}.
\item From the point of view of statistical physics, the regime on which Theorem \ref{thm:LocalisationFPU} is proved is not significant as the specific energy of the system $\varepsilon \sim 1/N^4$. The thermodynamic limit would require $\varepsilon$ to be constant and independent of the size of the system. This is read, in terms of the normal form construction, as a zero-dispersion limit of the Korteweg-de Vries equation. It would be interesting to study the effect of this limit.
\item Last, small attention has been given to the analysis of the FPU model when the dispersion is neglected (see \cite{PRK-PRE}). An interesting question to address would be if equations \eqref{eq:HamFPUNoDisp}-\eqref{eq:LastNoDisp} can be used to explain some properties of the dynamics, especially for short time scales, low Fourier modes or in the regime of high specific energy.
\end{itemize}

\begin{acknowledgement}
The authors are indebted to D. Bambusi and B. Rink for the interesting discussions which took place in the course of many years of fruitful interactions and collaborations on the subject. 
\end{acknowledgement}

\bibliographystyle{plain}
\bibliography{biblio_chapter}

\begin{thebibliography}{00}

\bibitem{ArnoldBookCM} V. Arnold, Mathematical Methods of Classical Mechanics, Graduate text in Mathematics, Springer, 1997

\bibitem{BBHM} P. Baldi, M. Berti, E. Haus, R. Montalto, Time quasi-periodic gravity water waves in finite depth, Inventiones Mathematicae, 214 (2), 739-911, 2018

\bibitem{BamPisa} D. Bambusi, Galerkin averaging method and Poincar\'e normal form for some quasilinear PDEs,
Ann. Scuola Norm. Sup. Pisa {\bf IV}, 669-702 (2005).

\bibitem{BamWW} D. Bambusi, Hamiltonian Studies on Counter-Propagating Water Waves, Water Waves {\bf3}, 49-83 (2021).

\bibitem{Bambusi2015} Bambusi D., Carati A., Maiocchi A. and Maspero A., Some analytic results on the FPU paradox (2015) Fields Inst. Commun. {\bf75} 235–54.

\bibitem{BCP} D. Bambusi, A. Carati and A. Ponno, The nonlinear Schr\"odinger equation as a resonant normal form, DCDS-B \textbf{2}, 109-128 (2002).

\bibitem{BamGre} D. Bambusi and B. Gr\'ebert, Birkhoff Normal Form for Partial Differential Equations with Tame Modulus,
Duke Math. Jour. {\bf135} 507-567 (2006).

\bibitem{BamMasToda} D. Bambusi, A. Maspero:  Birkhoff coordinates for the Toda Lattice in the limit of infinitely many particles with an application to FPU. J. Funct. Anal. 270 (2016), {\bf5}, 1818–1887

\bibitem{BamNek1} D. Bambusi and N.N. Nekhoroshev, A property of exponential stability in nonlinear wave equations near the fundamental linear mode, Physica D {\bf 122}, 73-104 (1998).

\bibitem{BamNek2} D. Bambusi and N.N. Nekhoroshev, Long Time Stability in Perturbations of Completely Resonant PDE's,
Acta Appl. Math. {\bf 70}, 1-22 (2002).

\bibitem{BamPo} D. Bambusi and A. Ponno, Comm. Math. Phys. \textbf{264}, 539-561 (2006).

\bibitem{PonBam08} D. Bambusi, A. Ponno, 
Resonance, Metastability and Blow-up in FPU
in "The Fermi-Pasta-Ulam problem", ed. by G. Gallavotti,
Springer Lecture Notes in Physics {\bf728}, 2008, 191-205

\bibitem{BLP} G. Benettin, R. Livi and A. Ponno,  The Fermi-Pasta-Ulam Problem: Scaling Laws vs. Initial Conditions
J. Stat. Phys. \textbf{135} (2009), 873-893.

\bibitem{Berti2021} M. Berti, L. Franzoi, A. Maspero, Pure gravity traveling quasi-periodic water waves with constant vorticity, arXiv:2101.12006

\bibitem{BertiFranzoiMaspero} M. Berti, L. Franzoi, A. Maspero, Traveling quasi-periodic water waves with constant vorticity,  ARMA {\bf240}(1):99-202, 2021.

\bibitem{Bous} J. Boussinesq, Essai sur la theorie des eaux courantes, Memoires presentes par divers savants ` l’Acad. des Sci. Inst. Nat. France, XXIII, 1877.

\bibitem{CS94} W. Craig and M. D. Groves, Hamiltonian long-wave approximations
to the water-wave problem, Wave Motion, {\bf19}(4) 367–389,
1994.

\bibitem{CS93} W. Craig and C. Sulem. Numerical simulation of gravity waves. J. Comput. Phys., {\bf108}(1) 73–83, 1993.

\bibitem{DNF} B.A. Dubrovin, S.P. Novikov and A.T. Fomenko, Modern Geometry -- Methods and Applications, part I, 
Springer-Verlag, 1992.

\bibitem{FeolaGiuliani} R. Feola, F. Giuliani, Quasi-periodic traveling waves on an infinitely deep fluid under gravity, arXiv:2005.08280

\bibitem{FPU1955} Fermi E., Pasta J. and Ulam S.,  Studies of non linear problems Los-Alamos Internal Report, 1955 Document LA-1940 first published in: Enrico Fermi Collected Papers, vol II, The University of Chicago Press, Chicago, and Accademia Nazionale dei Lincei, Roma, 1965, pp 977–988.


\bibitem{Gallavotti-2008} G. Gallavotti (ed), The Fermi–Pasta–Ulam Problem: A Status Report (Lecture Notes in Physics vol 728) (Berlin: Springer) 2008. 

\bibitem{GalloneLM}  M. Gallone, Hydrodynamics of the Fermi-Pasta-Ulam model and its integrable aspects, Master Thesis (2015)

\bibitem{GallonePonnoRuffo} M. Gallone, M. Marian, A. Ponno, S. Ruffo, Burgers turbulence in the Fermi-Pasta-Ulam-Tsingou chain, preprint

\bibitem{GaPa} M. Gallone and S. Pasquali, Metastability phenomena in two-dimensional rectangular lattices with nearest-neighbour interaction, Nonlinearity \textbf{34} 4983 (2021).

\bibitem{GPR} M. Gallone, A. Ponno and B. Rink, Korteweg-de Vries and Fermi-Pasta-Ulam-Tsingou: asymptotic integrability of quasi unidirectional waves
J. of Phys. A: Math Theor. \textbf{54}, 305701/1-29 (2021).

\bibitem{Gard} C.S. Gardner, Korteweg-de Vries Equation and Generalizations. IV. The Korteweg-de Vries Equation as a Hamiltonian System, J. Math. Phys. {\bf 12}, 1548-1551 (1971).

\bibitem{GelFom} I.M. Gelfand and S.V. Fomin, Calculus of Variations, Dover, 2000. 

\bibitem{Gus} F.G. Gustavson, On Constructing Formal Integrals of a Hamiltonian System Near an Equilibrium Point, The Astron. Jour. {\bf 71}, 670-686 (1966).

\bibitem{Kodama2} Y. Hiraoka and Y. Kodama, Normal Form and Solitons, in A.V. Mikhailov (ed.), Integrability, LNP
767, Springer 2009, 175-214.

\bibitem{PelinovskiHir} N. Hristov, D.E. Pelinovsky, Justification of the KP-II approximation in dynamics of two-dimensional FPU systems, arXiv:2111.03499

\bibitem{KaPo} T. Kappeler and J. P\"oschel, On the periodic KdV equation in weighted Sobolev spaces,
Ann. I. H. Poincar\'e -- AN 26, 841-853 (2009).

\bibitem{Kod85} Y. Kodama. Normal forms for weakly dispersive wave equations.
Phys. Lett. A, {\bf112}(5) 193–196, 1985.

\bibitem{Kod87a} Y. Kodama. Normal form and solitons. In Topics in soliton theory and exactly solvable nonlinear equations (Oberwolfach, 1986),
pages 319–340. World Sci. Publishing, Singapore, 1987.

\bibitem{Kod87b} Y. Kodama. On solitary-wave interaction. Phys. Lett. A,
{\bf123}(6) 276–282, 1987.

%\bibitem{Kuk} S. Kuksin, COSA???

\bibitem{Lax} P.D. Lax, Periodic Solutions of the KdV Equation, Comm. Pure Appl. Math. XXVIII,
141-188 (1975).

\bibitem{MaRa} J.E. Marsden and T.S. Ratiu, Introduction to Mechanics and Symmetry, 2nd ed., Springer-Verlag, 1999.

\bibitem{Moser} J.K. Moser, Lectures on Hamiltonian Systems, Memoirs of the Am. Math. Soc. {\bf 81},  1968.

\bibitem{Nek} N.N. Nekhoroshev, Strong Stability of the Approximate Fundamental Mode of the Nonlinear String Equation,
Trans. Moscow Math. Soc. { 2002}, 151-217.

\bibitem{Noe} E. N\"other, Invariant Variation Problems, English version reprinted in Transport Theory and Statistical Physics {\bf1}, 186-207 (1971).


\bibitem{PRK-PRE} P. Poggi, S. Ruffo, and H. Kantz, Shock waves and time scales to reach equipartition in the Fermi-Pasta-Ulam model, Phys. Rev. E{\bf 52}, 307 (1995)

\bibitem{PoBa} A. Ponno and D. Bambusi, Korteweg-de Vries equation and energy sharing in Fermi-Pasta-Ulam
Chaos \textbf{15} (2005), 015107/1-5


\bibitem{Quantum} A. A. Sokolov, J. M. Ternov, V. Ch. Zhukovskii, A. V. Borisow, Quantum electrodynamics, MIR Publishers Moscow 1988.

\bibitem{Shepelyansky} D.L. Shepelyansky, Nonlinearity {\bf10}, 1331 (1997).

\bibitem{Vaism_Book} Vaisman, I. Lecture on the Geometry of Poisson Manifolds. Progress in Mathematics {\bf118}, Basel: Birkhäuser Verlag, 1994

\bibitem{WS} G. Schneider and C. E. Wayne. The long-wave limit for the water wave problem. I. The case of zero surface tension. Comm. Pure Appl. Math., {\bf53}(12) 1475–1535, 2000.

\bibitem{ZK} N.J. Zabusky and M.D. Kruskal, Interaction of ``Solitons'' in a Collisionless Plasma and the Recurrence of Initial States, Phys. Rev. Lett. {\bf 15}, 240-243 (1965). 

%\bibitem{ZakSha} Zakharov and Shabat: NLS 

\bibitem{Zak68} V.E. Zakharov, Stability of periodic waves of finite amplitude on
the surface of a deep fluid, J. Appl. Mech. Tech. Phys., {\bf9} 190–194,
1968.

\bibitem{ZakBou} V.E. Zakharov, On stochastization of one-dimensional chains of nonlinear oscillators,
Sov. Phys. JETP {\bf 38},  108-110 (1974).

\bibitem{Zorn} M. Zorn,  Derivatives and Fr\'echet differentials, Bull. of the Am. Math. Soc., {\bf52}, 133-137 (1946).

\end{thebibliography}

\end{document}